\documentclass[prodmode,acmtecs]{acmsmall} 
\pdfoutput=1
\usepackage{epsfig}
\usepackage{color}
\usepackage[dvipsnames]{xcolor}
\usepackage{subfigure}
\usepackage{amsmath}
\usepackage[linesnumbered,ruled,vlined]{algorithm2e}
\renewcommand{\algorithmcfname}{ALGORITHM}
\SetAlFnt{\small}
\SetAlCapFnt{\small}
\SetAlCapNameFnt{\small}
\SetAlCapHSkip{0pt}
\IncMargin{-\parindent}

\usepackage[width=.75\textwidth]{caption}
\usepackage{lineno}

\acmVolume{x}
\acmNumber{x}
\acmArticle{x}

\renewcommand{\ae}[1]{{\color{black}{#1}}}

\newcommand{\re}[1]{{\color{black}{#1}}}

\setcopyright{rightsretained}

\doi{0000001.0000001}
\issn{1234-56789}

\begin{document}
\setcounter{page}{1}

\markboth{X. Gu et al.}{Efficient Schedulability Test for Dynamic-Priority Scheduling of MC Real-Time Systems}

\title{Efficient Schedulability Test for Dynamic-Priority Scheduling of Mixed-Criticality Real-Time Systems}
\author{Xiaozhe Gu
\affil{Nanyang Technological University}
Arvind Easwaran
\affil{Nanyang Technological University}
}

\begin{abstract}
Systems in many safety-critical application domains are subject to certification requirements. In such a system, there are typically different applications providing functionalities that have varying degrees of criticality. Consequently, the certification requirements for functionalities at these different criticality levels are also varying, with very high levels of assurance required for a highly critical functionality, whereas relatively low levels of assurance required for a less critical functionality. Considering the timing assurance given to various applications in the form of guaranteed budgets within deadlines, a theory of real-time scheduling for such multi-criticality systems has been under development in the recent past. In particular, an algorithm called Earliest Deadline First with Virtual Deadlines (EDF-VD) has shown a lot of promise for systems with two criticality levels, especially in terms of practical performance demonstrated through experiment results. In this paper we design a new schedulability test for EDF-VD that extend these performance benefits to multi-criticality systems. We propose a new test based on demand bound functions and also present a novel virtual deadline assignment strategy. Through extensive experiments we show that the proposed technique significantly outperforms existing strategies for a variety of generic real-time systems.

\end{abstract}







\maketitle

\section{Introduction}
\label{sec:intro}
Real-time systems are defined as those systems in which the correctness of the system depends not only on the logical result of computation, but also on the time at which the results are produced~\cite{stankovic1998introduction}.  For example  a pacemaker is inserted in a person's chest to provide electrical impulses at regular intervals to help the heart beat. Here the pacemaker must provide service with certain timing constraints, and  applications with these kinds of timing constraints are considered real time.

Timing constraints in real-time systems are often modeled as deadlines. If a schedulable activity (e.g., a piece of job) executes and completes before its assigned deadline, the deadline is met (and otherwise it is missed). This means, in order to meet the deadline, the scheduler must have apriori knowledge on the amount of execution that the job would request. Further, to achieve high timing predictability in the presence of various sources of variability, these systems must be built under pessimistic assumptions to cope with worst case scenarios. That is, the scheduler typically assumes that the job would execute for a certain worst-case amount of time (denoted as WCET for Worst-Case Execution Time), which encompasses all the possible variations in execution time.   However determining an exact  WCET for a job is very difficult~\cite{Puschner:2000}, and usually a conservative overestimation of the true WCET~\cite{Wilhelm:2008} is used to analyze and schedule a real-time system.

An increasing trend in embedded systems that multiple functionalities with different levels of ``criticality" (or importance) are developed independently and integrated together on a common computing platform~\cite{prisaznuk1992integrated}. This trend is evident in industry-driven initiatives such as ARINC653~\cite{ARINC653}, Integrated Modular Avionics (IMA) in avionics and AUTOSAR in automotive. An important notion behind this trend is the safe isolation of separate functionalities of different criticality, primarily to achieve fault containment. 
For example in a modern car, devices and software are integrated into the entertainment system and often run on the same platform as the more critical instrument panel information display subsystem. The challenge in such a system is managing the dramatically-different nature of resource requirements between the  less critical infotainment functions characterized by best-effort or soft real-time needs and the more critical display functions that require strong reliability.

Here we define criticality as the level of assurance against severe failure needed for a system component. Typically, failure of high critical functionality   would cause a more severe consequence to the whole system than the failure of a low critical functionality.
Thus high critical functions require higher assurance that their estimated WCETs will not be exceeded. As a result, their  WCET estimates tend to involve  very conservative assumptions about the system (cache flushing on preemption, over-provisioning for potentially missed execution paths, etc.) that are very unlikely to occur in practice. Consequently, the system resources are in fact severely under-utilized in practice because high critical functions would rarely execute as much as their WCET estimates.

This raises the challenge of how to balance the conflicting requirements of isolation for safety assurance and efficient resource sharing for economical benefits. The concept of \textit{mixed-criticality} (MC) appears to be important in meeting those two goals.  In order to close such a gap in resource utilization, Vestal~\cite{Ves07} proposed the MC task model that comprises of different WCET estimates. These estimates  are determined at different levels of confidence (``criticality'') based on the following principle. A reasonable low-confidence WCET estimate, even if it is based on measurements, may be sufficient for almost all possible execution scenarios in practice.   As long as this estimate is not violated, both low critical tasks and high critical tasks are required to meet their timing constraints.  In the highly unlikely event that this estimate is violated, as long as the scheduling mechanism can ensure deadline satisfaction for high critical applications, the resulting system design may still be considered as safe. Considering such MC real-time task systems with two criticality levels, several studies have proposed scheduling algorithms and corresponding schedulability tests in the past~\cite{EkYi12,BBA12,BBD11}. There have also been some recent studies that extend some of these results to more than two criticality levels~\cite{EkYi14,fleming2013extending,BBD15}.

In this paper we focus on the problem of EDF (earliest deadline first) scheduling of mixed-criticality systems on uniprocessors. In particular, we address the problem of scheduling multi-criticality real-time task systems (systems with more than two criticality levels). Baruah and Vestal~\cite{BaVe08} first considered MC scheduling with EDF. Later Park and Kim~\cite{park2011dynamic} proposed Criticality Based EDF that applies a combination of off- and on-line analyses to run high critical jobs as late as possible, and low critical jobs in the generated slack.    Baruah et al. proposed an algorithm called EDF-VD (EDF - with virtual deadlines)~\cite{BBA12} for a dual-criticality system. High critical tasks have their deadlines reduced  by the same factor (if necessary) during low-criticality mode execution.  They demonstrate both theoretically and via evaluations that this is an effective scheme.   EDF-VD~\cite{BBA12} is constrained to dual-criticality implicit deadline  systems, and was extended  to multi-criticality  systems in  \cite{BBD15}.   Although EDF-VD analyses the system across multiple criticality levels together, it is still very pessimistic even for two criticality levels because of the following factors: 1) The virtual deadlines for all the high critical tasks are uniformly assigned based on a single common factor, and 2) Using demand density to characterize the demand of a constrained-deadline task (task with a virtual deadline smaller than period) will always be pessimistic.

A more general demand bound function based analysis for EDF mixed-criticality scheduling was proposed by Ekberg and Yi~\cite{EkYi12}. They also introduced a heuristic virtual deadline tuning algorithm called GreedyTuning  where deadlines can be reduced by different factors. GreedyTuning can increase the chances that a task system is schedulable by EDF. From extensive experiments, they show that GreedyTuning outperforms all existing works on MC scheduling for a variety of generic real-time systems.  GreedyTuning was extended to multi-criticality systems in~\cite{EkYi14}. However it suffers from a drawback that its schedulability performance  drops significantly as the number of criticality levels increases~(see Figure 17 in~\cite{EkYi14}).  In Figure~\ref{fig:L2s} of this paper,  we also show that it is not good at scheduling task systems with a larger percentage of high critical tasks. The primary reason for this drop in performance is that they analyze the system in each criticality mode from the time instant when the mode-switch happens, conservatively assuming maximum carry-over interference when the system behavior switches from a lower criticality level to the one being analyzed. 

The test proposed in \cite{Eas13} addresses this problem by considering the system behavior from the start of a busy interval in a dbf-based analysis, but it is restricted to dual-criticality systems.    The first challenge in extending this initial result to more than two criticality levels is that the task execution pattern that can result in worst-case demand and hence the dbf for dual-criticality system is no longer valid in multi-criticality systems.   Besides, if we consider the demand from the start of a busy interval as \cite{Eas13} does, there would be  multiple mode-switches happening at $S_1,~S_2,\ldots,~S_m$  during the time interval. Then, we have to consider  all possible combinations of these mode-switch instants, and as a result,  the complexity of the test is exponential in the number of criticality levels. Also, given a set of these mode-switch instants, the task execution pattern that will result in the worst-case demand depends on all the tasks in the system and their remaining execution time at each of those mode-switches.  It is therefore non-trivial to determine a worst-case pattern with low pessimism, and there is no known technique for the same.

\textbf{Contribution:} In sum the contributions of this work can be summarized as follows.
\begin{itemize}
	\item In this paper, we overcome the challenges and extend the dual-criticality dbf-based test~\cite{Eas13} for multi-criticality systems that have the same time complexity as the dual-criticality ones. 
	\item  To further improve the performance of the proposed design,  we also
		develop a new virtual deadline assignment strategy, extending the strategy proposed by ~\cite{EkYi14}.
	Finally, through experimental evaluation, we demonstrate that the proposed technique significantly outperforms the existing ones for generic real-time MC task systems.
\end{itemize}

\textbf{Other Related Work}: Since Vestal~\cite{Ves07} first proposed the MC task model and an algorithm based on Audsley's priority assignment strategy~\cite{Aud91}, there have been a series of publications on the scheduling of MC systems on uniprocessors. A number of proposed studies are restricted to address the problem of scheduling a finite set of mixed-criticality jobs with criticality dependent execution times~\cite{LiBa10a,BLS10a,BBD12}. The model these studies use is a constrained one, because in many real-time systems each task is able to generate an infinite number of jobs.  For example, the engine control unit in a car periodically senses and processes information to efficiently control the fuel injection and emissions.  Hence, these studies are superseded by  studies that are applicable to the more general sporadic MC task model, which is also the focus of this paper.  Vestal's approach~\cite{Ves07} is the first work that uses  Response-Time Analysis (RTA) to analyze the schedulability of generic MC task systems. This work was later improved by  the  Static Mixed Criticality Scheme (SMC)~\cite{BaBu11}. Adaptive scheme (AMC)~\cite{BBD11}  goes further and  it outperforms all the previous works on fixed priority MC scheduling in terms of schedulability. Fleming and Burns~\cite{fleming2013extending}  extended AMC for task systems  with an arbitrary number of criticality levels, focusing particularly on five levels as this is the maximum found in automotive and avionics standards. There are various works~ (e.g.~\cite{Su13,BFB15,bate2015bailout}) concerned about addressing other problems about MC scheduling, e.g., how to switch back to low critical mode or support low critical execution, but is not the focus of this work.

The rest of the paper is organized as follows. In Section~\ref{sec:model}  we  first introduce the mixed-criticality task and scheduling model, and in Section~\ref{sec:ExistDBF} we  give a brief introduction about  GreedyTuning~\cite{EkYi14}, which is also the work we aim to improve.  We derive our new multi-mode demand bound function (dbf) in Section~\ref{sec:GIT}. In Section~\ref{sec:T1} we present a novel deadline tuning algorithm to improve the performance of the proposed test. Finally in Section~\ref{sec:eva} we show that our proposed test  dominates GreedyTuning~\cite{EkYi14}  from experimentation.

\section{Task and Scheduling Model}
\label{sec:model}
\subsection{MC Task Model}
The  sporadic task model~\cite{Mok1983} is a generic model for capturing the
real-time requirements of many event-driven systems including those with MC such as avionics and automotive. A  sporadic task~\cite{Mok1983} can be specified as $\tau_i=(T_i,C_i,D_i)$, where $T_i$ denotes minimum separation between successive job releases, $D_i$ denotes  its relative deadline, and $C_i$ denotes its worst-case execution time (WCET). No job of $\tau_i$ is expected to execute for more  time than its WCET, and otherwise the system is regarded as exhibiting erroneous behavior. Any job  released by $\tau_i$ is required to complete by its deadline, and deadline miss is regarded as system failure.

The task model widely used in most previous studies on  scheduling of MC systems (e.g.,~\cite{EkYi12,BBD11,Ves07,Eas13}) is a straightforward  extension of the classic sporadic task model~\cite{Mok1983} to a MC setting; the worst-case execution times of a single task can vary between criticality levels. However in most of these works, the task model  is constrained to two levels, i.e, a task can either be a low critical task or a high critical task.  Instead, we use a more general multi-criticality model in this paper.
Formally, a task $\tau_i$ in a MC sporadic task set  $\tau=\{\tau_1, \tau_2,\ldots,$ $\tau_k\}$  can be represented as a tuple $(T_i,C_i,D_i,L_i)$,  where:
\begin{itemize}
\item { $T_i \in \mathbb Z^+$ is the minimal time separation between the release of two successive jobs\footnote{We focus on integer release time model and hence $T_i \in \mathbb Z^+$.},}
\item {$L_i  \in \mathbb Z^+$ is the criticality level of  $\tau_i$, and $L_i=1$ denotes the lowest criticality level, }
\item {$D_i\in \mathbb Z^+$ is the relative deadline},
\item {$C_i=(C_i^1,~C_i^2,\ldots,C_i^{L_i})$ is a $L_i$-tuple of estimated execution time budgets respectively for each criticality level. }
\end{itemize}

Since the worst-case execution times for higher criticality levels are estimated more conservatively, we make the  standard assumption that
\[
\forall \tau_i\in \tau:~ C_i^1 \leq C_i^2 \ldots \leq C_i^{L_i} \leq D_i \leq T_i
\]

\subsection{Mode Switch}
The system initially starts in $L_1$ (short for level one) criticality mode (i.e., the lowest criticality mode), each task $\tau_i\in \tau$ releases a  potentially infinite sequence of jobs $<J_i^1,J_i^2,\ldots>$ in the standard manner: if $r(J),~d(J)\in \mathbb R $ denote the release time and absolute deadline of job $J$, then 
\begin{itemize}
\item $r(J_i^{k+1})\geq r(J_i^k)+T_i$
\item $d(J_i^k)=r(J_i^k)+ D_i$
\end{itemize}
The system stays in  $L_1$ criticality mode as long as all the jobs of every  task $\tau_i$ with $L_i\geq 2$ do not execute beyond their $L_1$   execution time estimate $C_i^1$.  Once a  job executes for its entire $L_1$ execution time estimate $C_i^1$ without signaling that it has finished, the system immediately switches to $L_2$ criticality mode. That is, in general as shown in Figure~\ref{FigureModeInstant}, when the  first job $J_i$ such that $L_i \geq m$  executes for more than  its $L_{m-1}$  execution time estimate $C_i^{m-1}$ but does not signal that it has finished, the system switches to $L_{m}$ criticality mode.  The mode switch time instant when the system switches from $L_{m-1}$ to $L_m$ criticality mode is called \emph{criticality mode switch instant} and is denoted as $S_{m}$. 
\begin{figure}[tbh]
\begin{center}
\noindent
  \includegraphics[width=0.7\textwidth]{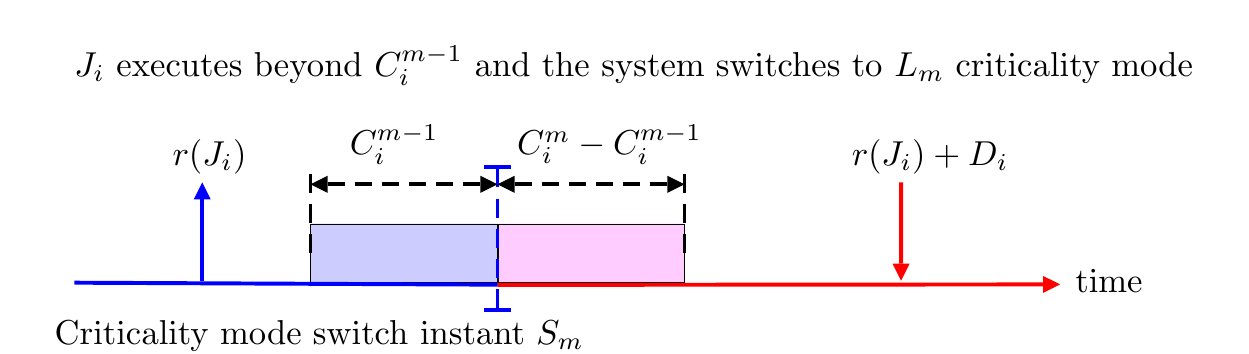}
  \end{center}
    \caption{Criticality Mode Switch Instant $S_m$}\label{FigureModeInstant}
\end{figure}

After this mode switch at $S_m$, jobs with criticality level $L_i<m$, including those that were released before $S_m$ but did not complete until $S_m$,  are no longer required to meet deadlines. {For simplicity of analysis, these jobs are assumed to be dropped thereafter.} However we must still meet all the deadlines for jobs with $L_i\geq m$, even if they require up to $C_i^{m}$ budgets.  If the system is now in $L_m$ criticality mode, we assume it will switch back to $L_{1}$  mode whenever the processor is idle. 

Note that the assumptions in the preceding paragraph on mode switch are consistent with the standard literature on MC scheduling (e.g.,~\cite{LiBa10a,BBD11,Ves07,Eas13}). There are some  studies that focus on dealing with those assumptions, such as reducing the penalty on low-critical tasks (e.g.,~\cite{PCH14,PCH13,GU15}), and switching back to a lower criticality mode earlier than the processor idle time (e.g.,\cite{DBLP:conf/ecrts/BateBD15}).  These studies are orthogonal to the focus of this paper, which is to derive efficient schedulability tests for EDF-scheduled MC task systems. 

\begin{definition}[MC-Schedulable]
We define a MC task system to be MC-schedulable if $\forall m: 1 \leq m\leq M$ where $M=\max_{\tau_i\in \tau}\{L_i\}$, all jobs with criticality level $L_i\geq m$ can receive a budget up to $C_i^{m}$ and signal completion between their release time and deadline while the  system stays in $L_{m}$ criticality mode.
\end{definition}

\subsection{Demand Bound Function Based Schedulability Analysis}
Demand bound function  was first proposed to analyze the schedulability of non-MC real-time workloads~\cite{BMR90}.  The demand bound function captures the maximum execution demand a task can generate for a given time interval length. 

\begin{definition}[Demand bound function]
A demand bound function $dbf(\tau_i,\mathit e)$ gives an upper bound on the maximum possible execution demand of task $\tau_i$ in any time interval of length $\mathit e$, where demand is calculated as the total amount of required execution time of jobs with their whole scheduling windows within the time interval.
\end{definition}
For example  a task $\tau_i=(T_i=5,C_i=2,D_i=3,L_i=1)$ can generate as much  as $2\times C_i=2\times 2=4$   time units execution demand for a time interval length equal to $8$.
For non-MC constrained deadline task model, the demand bound function for a given time interval length $\mathit e$  can be computed in constant time~\cite{BMR90} using the following equation.
\begin{equation}
dbf(\tau_i,\mathit e)=\left(\left\lfloor  \frac{ \mathit e -D_i }{T_i} \right\rfloor +1\right)\times C_i
\end{equation}

As long as we can guarantee that the total execution demand of a task set $\tau$ is always smaller than or equal to the time interval length $\mathit e$ for all values of $\mathit e$, we can claim $\tau$ is schedulable by the EDF algorithm on a uniprocessor platform. 
\begin{theorem}[\cite{BMR90}]
\label{theorem:test_non}
A non-mixed-criticality sporadic task set $\tau$ is successfully scheduled by the earliest deadline first (EDF) algorithm on a dedicated unit speed uniprocessor platform if
\[
		\forall \mathit e \in\{1,2,\ldots,\mathit e^{max}\}:~\sum_{\tau_i\in\tau}dbf(\tau_i,\mathit e)\leq  \mathit e
\]
where $\mathit e^{max}$  is a pseudo-polynomial in the size of the input~\cite{BMR90} as long as the system utilization is bounded by some constant smaller than 1.
\end{theorem}


\section{Background: Existing DBF based Test for MC Task Systems}
\label{sec:ExistDBF}
In this section, we first extend the idea of demand bound function to the mixed-criticality setting. Then we  introduce the Single-Mode (SM) demand bound function  of mixed-criticality workloads derived in an existing work~\cite{EkYi12}. Test based on this SM  dbf   has been shown to dominate previous studies (e.g., EDF-VD~\cite{BBA12} and AMC-max~\cite{BBD12}) in terms of the ability to schedule MC sporadic task systems.


Let $M=\max_{\tau_i\in \tau}\{L_i\}$ and  $dbf_{SM}(\tau_i,\mathit e,m)$ where $m\in \{ 1,2,\ldots, M \}$ denote the SM demand bound function of $\tau_i$ for the time interval $[S_m,S_m+\mathit e)$, when the system is currently in $L_m$ criticality mode and was in $L_{m-1}$ criticality mode before that. As we can observe, SM dbf test separately analyzes the system in each criticality mode because it only considers the demand during $[S_m,S_m+\mathit e)$ but the system behavior before $S_m$ is totally ignored. Therefore,   Theorem~\ref{theorem:test_non} can be extended in a  straightforward way as follows.
\begin{theorem}[Proposition 3~\cite{EkYi14}]
\label{theorem:theorem0}
A MC task set $\tau$ is schedulable by EDF on a dedicated unit speed uniprocessor platform  for all the criticality modes if the following conditions hold:
\begin{align*}
&\forall~m\in\{1,2,\ldots,M\}:~\forall \mathit{e}\in\{1,2,\ldots,\mathit e^{max}\}:~dbf_{SM}(\tau,\mathit{e},m)\leq \mathit{e}
\end{align*}
where $dbf_{SM}(\tau,\mathit{e},m)=\sum\limits_{L_i\geq m}dbf_{SM}(\tau_i,\mathit{e},m)$, and $\mathit e^{max}$  is a pseudo-polynomial in the size of the input \re{if the utilization of each criticality mode is bounded by some constant smaller than 1}. 
\end{theorem}
We define condition $\mathbf{CN_m^S}$ as follows
\begin{align*}
 \forall~ \mathit{e}\in\{1,2,\ldots,\mathit e^{max}\}:~\sum_{L_i\geq m}dbf_{SM}(\tau_i,\mathit{e},m)\leq \mathit{e}
\end{align*}


Condition $\mathbf{CN_m^S}$ captures the schedulability of task set $\tau$ for $L_m$ criticality mode on the assumption that $\tau$ is schedulable in $L_{m'}$ ($m'<m$) criticality modes.  
To compute $dbf_{SM}(\tau_i,\mathit{e},m)$ we need to determine the maximum demand that $\tau_i$ can generate in the interval $[S_m,S_m+\mathit e]$. However, the demand of any task $\tau_i$ in  $L_m$ criticality mode depends on the release pattern in all the previous criticality modes. 

In $L_1$ criticality mode, each task $\tau_i$ will behave like a normal non-mixed-criticality task, and all jobs are guaranteed to execute for at most $C_i^1$ time units. Therefore the dbf for non-mixed-criticality tasks  can be directly applied to capture the demand of $\tau_i$ in $L_1$ criticality mode, i.e., 
\begin{equation}
\label{eqn:dbf^0_task_1}
dbf_{SM}(\tau_i,\mathit e,1)=\left(\left\lfloor  \frac{ \mathit e -D_i}{T_i} \right\rfloor +1\right)\times C_i^1
\end{equation}

If $\tau_i$ has $L_i<m$, then  $dbf_{SM}(\tau_i,\mathit{e},m)=0$ because tasks with $L_i<m$ will be discarded  after $S_m$. On the other hand, if $L_i\geq m$, then we need to consider the job that is released before the mode switch instant $S_m$ but has its deadline after $S_m$, because this job can affect  $\tau_i$'s execution demand after $S_m$. We call such jobs as carry-over jobs.
\begin{definition}[Carry-over job]
A job $J_i^m$ from a $L_i\geq m$ criticality task that is active (released before $S_m$ and has deadline after $S_m$) at the time of the switch to $L_m$ criticality mode is called a carry-over job for $L_m$ criticality mode. 
\end{definition}

\subsection{Characterizing the Demand of Carry-Over Jobs}
While we can discard all the active jobs with $L_i~(L_i<m)$,  the remaining execution demand of the carry-over jobs must be completed by their respective deadlines, and hence the demand of carry-over jobs must be accounted for in $dbf_{SM}(\tau_i,\mathit{e},m)$. 
\begin{figure}[tbh]
\begin{center}
\noindent
  \includegraphics[width=0.85\textwidth]{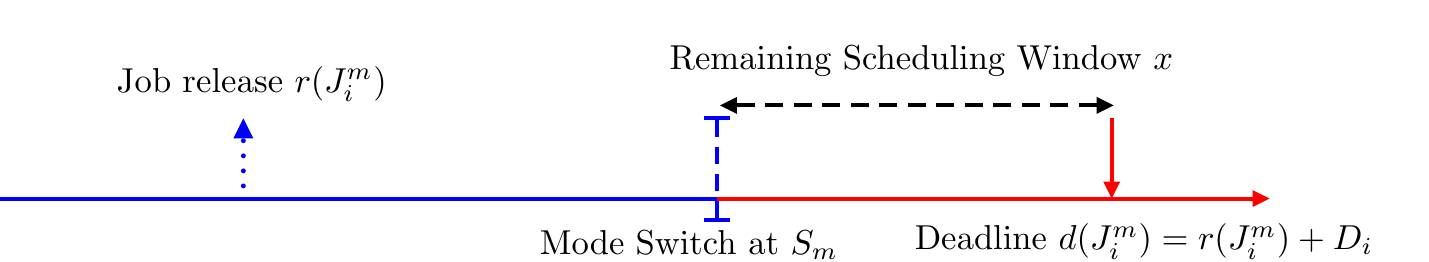}
  \end{center}
\caption{After a switch to higher criticality mode, the remaining execution demand of a carry-over job must be finished in its remaining scheduling window.}
\label{fig:dbfc1}
\end{figure}




At the time of the switch to $L_m$ criticality mode, a carry-over job $J_i^m$ from task $\tau_i$ has $x~(x\geq 0)$ time units left until its deadline as shown in Figure~\ref{fig:dbfc1}.  Since this job would have met its deadline in $L_{m'}$ where $(m'<m)$ criticality mode if the mode-switch at $S_m$ had not happened, there can be at most $x$ time units left to finish its maximum possible remaining execution demand $C_i^{m-1}$.  That is, the job must  have already executed at least $\max(0,C_i^{m-1}-x)$ time units before the mode-switch at $S_m$ (otherwise deadline miss could happen in $L_{m-1}$ criticality mode).  After mode-switch $S_m$, the carry-over job may now execute for up to $C_i^m$ time units in total, and therefore the total execution demand remaining for the carry-over job after the switch is at most $C_i^m-\max(0,C_i^{m-1}-x)$. Unfortunately if $x\to 0$, condition $\mathbf{CN_m^S}$ with $(m>1)$ cannot be satisfied because as long as $C_i^m-C_i^{m-1} >0$, we can always find a small $x$ so that $dbf_{SM}(\tau_i,\mathit e\rightarrow 0,m
)=C_i^m-C_i^{m-1} >0$.

The problem described  above stems from  the fact that EDF may execute a carry-over job  quite late, and hence it can not finish its remaining execution demand after the mode-switch at $S_m$.  To solve this problem, virtual deadlines in different criticality modes have been introduced~\cite{BBA12,EkYi12,Eas13}.  When the system is in $L_m$ criticality mode, tasks $\tau_i$ with $L_i\geq m$ are scheduled by EDF scheduler according to virtual deadline $D_i^m$~$(D_i^m\leq D_i)$.  In $L_m$ criticality mode,  any task $\tau_i$ with $L_i\geq m$ must finish its execution demand $C_i^m$ by its virtual deadline $D_i^m$. Since a job $J_i$ can now have multiple deadlines, we use $d(J_i,m)=r(J_i)+D_i^m$  to denote its absolute virtual deadline for $L_m$ criticality mode.

Virtual deadline enables the carry-over job to have extra slack time, $D_i^m-D_i^{m-1}$, to finish its additional  demand $C_i^m-C_i^{m-1}$ at the cost of a higher load of execution demand in lower criticality mode. However we should note that virtual deadlines are not actual deadlines, and  can be determined by  deadline tuning algorithms~\cite{EkYi12,Eas13} to improve EDF schedulability. The remaining execution demand for a carry-over job can then be bounded with the following lemma.

\begin{lemma}[Demand of carry-over jobs~\cite{EkYi12}]
Assume that EDF uses virtual relative deadline $D_i^m$ in $L_m$ criticality mode  for tasks with $L_i\geq m$, and that we can guarantee that the demand is met in all lower criticality modes (i.e., $L_{m'}$ with $(m'<m)$) with respective virtual deadlines. If the carry-over job of $\tau_i$ has a remaining scheduling window of $x$ time units until it deadline $D_i^m$, as illustrated in Figure~\ref{fig:dbfc2}, then the following hold:
\begin{enumerate}
	\item If $x<D_i^m-D_i^{m-1}$, then the job must has already finished before $S_m$.
	\item If $x\geq D_i^m-D_i^{m-1}$, then the job's remaining execution demand after $S_m$ is bounded by $C_i^m-\max(0,C_i^{m-1}-x+D_i^m-D_i^{m-1})$.  
\end{enumerate}
\end{lemma}
As we can observe, to  maximize the total demand in $L_m$ criticality mode, SM  dbf  conservatively assumes maximum possible carry-over demand  $C_i^m-\max(0,C_i^{m-1}-x+D_i^m-D_i^{m-1})$ from  $L_{m-1}$ criticality mode.

\begin{figure}[tbh]
\begin{center}
\noindent
  \includegraphics[width=0.65\textwidth]{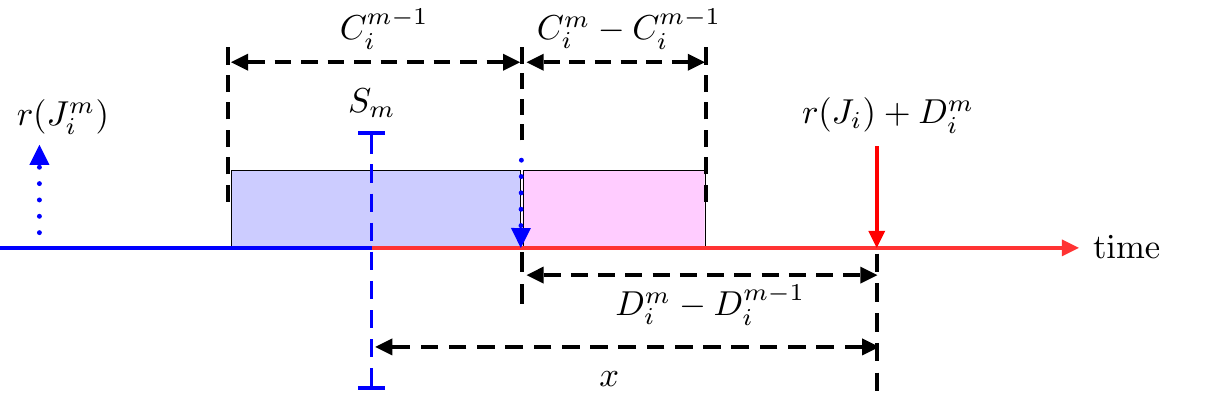}
  \end{center}
    \caption{A carry-over job of $\tau_i$ has a remaining scheduling window of length x after the switch to $L_m$ mode. Here the switch happens before $r(J_i^m)+D_i^{m-1}$.}\label{fig:dbfc2}
\end{figure}

\subsection{ Formulating the SM DBF}

The execution demand of $\tau_i$ for  time interval  $[S_m,S_m+\mathit e)$ is equal to sum of the unfinished execution demand of the carry-over job, and the demand of jobs released after the carry-over job in this interval. 
\begin{lemma}[Maximum demand pattern~\cite{EkYi14}]
\label{lemma:pt1}
Task $\tau_i$ with $L_i\geq m$ can generate maximum execution demand in $L_{m}~(m>1)$ criticality mode for a time interval length $\mathit e$  when the corresponding virtual deadline of some job $r(J_i)+D_i^m$ is at the end of this time interval $\mathit e$ and all preceding jobs are released as late as possible, as shown in Figure~\ref{fig:dbfc3}.
\end{lemma}
Therefore the demand bound function of $\tau_i$ in $L_{m}~(m>1)$ criticality mode can be summarized as follows:
\begin{equation}
dbf_{SM}(\tau_i,\mathit{e},m)=\max(0, (1+\lfloor \frac{ \mathit e -(D_i^m-D_i^{m-1})}{T_i} \rfloor )\times C_i^m)-done_m(\tau_i,\mathit e)
\end{equation}
Here $done(\tau_i,\mathit e)$ captures the execution demand of the carry-over job that must finish before $S_m$ and is equal to:
\begin{equation}
\label{eqn:done}
\begin{split}
   done_m(\tau_i,\mathit e)=
   \begin{cases}
   \max(0,C_i^{m-1}-mod(\mathit e,T_i)+D_i^m-D_i^{m-1}), &\mbox{if } D_i^m-D_i^{m-1}\leq mod(\mathit e,T_i)\\ &\mbox{ and }mod(\mathit e,T_i) <D_i^m \\
   0 &\mbox{otherwise}
   \end{cases}
\end{split}
\end{equation}
where $mod(\mathit e,T_i)=e~mod~T_i$ is equal to the length of the  remaining scheduling window of the carry-over job after $S_m$.

\begin{figure}[tbh]
\begin{center}
\noindent
  \includegraphics[width=0.65\textwidth]{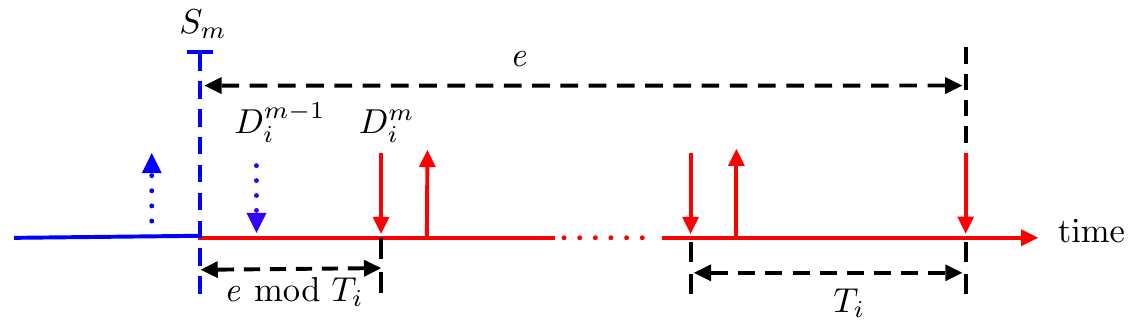}
  \end{center}
    \caption{Demand of $\tau_i$ in $L_{m}$~$(m>1)$ Criticality Mode}
    \label{fig:dbfc3}
\end{figure}

\section{Multi-Mode DBF Based Test For MC Task Systems}
\label{sec:GIT}
The SM demand bound function  considers the behavior of each criticality mode separately. That is, it does not use the execution demand in the previous criticality mode (e.g, $L_{m-1}$) to determine the remaining execution for carry-over jobs when mode switch happens at  $S_m$.  If the execution load of task set $\tau$ in $L_{m-1}$ criticality mode is low, then carry-over jobs of many tasks would finish well before their deadlines, and  would not generate any  carry-over demand.  This property leads to some interesting results like the fact that SM dbf based test  cannot even schedule some task sets that are schedulable by reservation based approaches (i.e., all tasks are allocated $C_i^{L_i}$ execution budgets). 


\begin{example}
Suppose task set $\tau=\{\tau_1,\tau_2\}$ has two tasks, where $\tau_1$ and $\tau_2$ are  given in  the following table.  Obviously $\tau$ is schedulable by reservation based approaches because  the utilization $\frac{C_1^2}{T_1}+\frac{C_2^2}{T_2}<1$.  However according to GreedyTuning~\cite{EkYi14}, $\tau$ is not schedulable because $dbf_{SM}(\tau,4,2)\leq 4\wedge dbf_{SM}(\tau,4,1)\leq 4$ is not true for any possible combination  of virtual deadlines.  We have to set small virtual deadlines in order to make $dbf_{SM}(\tau,4,2)\leq 4$.  However,  by doing so,  $dbf_{SM}(\tau,4,1)$  would be  greater than $4$.

\begin{table}[h]
\center
\begin{tabular}{|l|l|l|l|l|}
 \hline
Task& $T_i$& $C_i$&$D_i$&$L_i$  \\
 \hline
$\tau_1$& $15$& $\{C_1^1=3,C_1^2=7\}$  &$15$& $2$\\
$\tau_2$& $2$ & $\{C_2^1=1,C_2^2=1\}$ &$2$&$2$ \\
 \hline
\end{tabular}
\label{ETSt}
\end{table}
\end{example}



To address this drawback,  in this section  we propose a Multi-Mode (MM) demand bound function that collectively bounds the  demand of $\tau$ in  $L_{m-1}$ and $L_m$ criticality modes. Suppose the time interval of interest is $[S_{m-1},  S_{m-1}+\mathit e)$ where $S_{m-1}\leq S_m\leq S_{m-1}+\mathit e$. For simplicity  we assume $S_{m-1}=0$  because  what determines the total demand is the time interval length $S_m-S_{m-1}$ and $\mathit e$, i.e., the dbf is independent of the exact value of $S_{m-1}$. As a result, the interval becomes $[0,\mathit e)$ where $0\leq S_m\leq \mathit e$.  Let $dbf_{MM}(\tau,S_m,\mathit{e},m)$ denote the total execution demand of $\tau$ for the time interval $[0,\mathit e)$.

The test proposed in \cite{Eas13} considers the system behavior from the start of a busy interval in a dbf-based analysis, but is limited to a dual-criticality task system.   In a dual-criticality task system, there is at most one mode-switch. However, in a  task system with more than two criticality levels, we have to consider the demand from jobs released before $S_{m-1}$ but have deadline after $S_{m-1}$.  As a result, the dbf analysis in \cite{Eas13} is no longer valid in multi-criticality systems.  In this section we extend the dbf analysis in \cite{Eas13} to multi-criticality systems and   present a MM dbf based test for multi-criticality systems. 

The following theorem is a straightforward extension of Theorem~\ref{theorem:theorem0}, and can be used to determine whether a task set $\tau$ is schedulable by EDF on a dedicated unit speed uniprocessor platform.
\begin{theorem}
\label{theorem:MM}
A MC task set $\tau$ is schedulable by EDF on a dedicated unit speed uniprocessor platform  for all the criticality modes if the following conditions hold:
\begin{align*}
\forall~m\in\{1,2,\ldots,M\}:\forall~\mathit{e}\in\{1,2,\ldots,\mathit e^{max}\}:\forall S_m\in\{1,2,\ldots,\mathit e\}:~dbf_{MM}(\tau,S_m,\mathit{e},m)\leq \mathit e
\end{align*}
where $\mathit e^{max}$ is  pseudo-polynomial in the size of the input~\cite{GU15} if the system utilization of each criticality mode is bounded by some constant smaller than 1, and is also derived in Appendix~\ref{sec:emax}.
\end{theorem}


In the time interval of interest $[S_{m-1}=0,\mathit e)$, a job can experience at most two mode-switches at $S_{m-1}=0$ and $S_m$, respectively. We categorize such carry-over jobs into four types: $J_i^X$ where $X\in\{A,B,C,D\}$, and use $J_i^{X+1}$ to denote the next job released after $J_i^X$. The patterns of these jobs are shown in Figure~\ref{fig:spj}.  
\begin{description}
\item[$J_{i}^{A}$] $L_i=m-1$, $r(J_i^{A})<S_{m-1}$ and $r(J_i^{{A+1}})\geq S_{m-1}$.
\item[$J_i^{B}$] $L_i=m-1$, $r(J_i^{B})\geq S_{m-1}$, $r(J_i^{{B}})\leq S_m$ and $r(J_i^{{B+1}})>S_m$.
\item[$J_i^{C}$] $L_i\geq m$, $r(J_i^{C})<S_{m-1}$ and $r(J_i^{{C+1}})\geq S_{m-1}$.
\item[$J_i^{D}$] $L_i\geq m$, $r(J_i^{D})\geq S_{m-1}$, $r(J_i^{{D}})\leq S_m$ and $r(J_i^{{D+1}})> S_m$.
\end{description}
\begin{figure}[h]
\centering 
\includegraphics[scale=0.75]{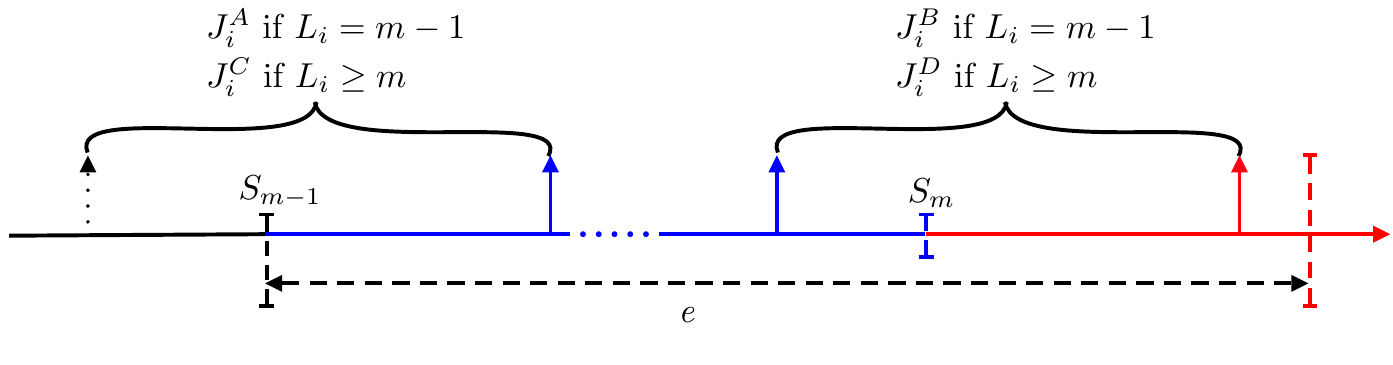} 
\caption{Four types of jobs that experience mode switch.}
\label{fig:spj}
\end{figure}
Let $dbf_{MM}(J_i^{X},S_m,\mathit e,m)$, where $X\in \{A,B,C,D\}$, to denote the  demand of job $J_i^{X}$ during $[0,\mathit e)$. Given the value of $r(J_i^X)$, we can simply calculate its demand. Thus  the detailed equations of $dbf_{MM}(J_i^{X},S_m,\mathit e,m)$ will be presented in Appendix~\ref{sec:dbfjob}  because it is very intuitive.

\subsection {MM dbf for a single task ($dbf_{MM}(\tau_i,S_m,\mathit{e},m)$)} 
In the time interval $[S_{m-1}=0,\mathit e)$,  tasks with $L_i=m-1$ could execute during $[0,S_m)$, but would be dropped after $S_m$. Tasks with $L_i\geq m$ could execute during the whole time interval $[0,\mathit e)$.  For $L_1$ criticality mode (the initial mode),   $dbf_{MM}(\tau_i,S_0,\mathit e,1)$  can be obtained from demand bound function for non-mixed sporadic tasks (see Equation~\ref{eqn:dbf^0_task_1}).
\subsubsection {Case when $L_i=m-1$}
The maximum demand generated by task $\tau_i$ with $L_i=m-1$ during the interval $[0,S_m)$ is equal to the sum of demand of all jobs released  during $[0,S_m)$ and the execution demand of $J_i^{A}$ during $[0,S_m)$.

Let $DEM_i(r_A)$ denote the demand that $\tau_i$ generates   during $[0,S_m)$ when $r(J_i^{A})=r_A$ and all jobs are released  as soon as possible with period $T_i$. Given $r(J_i^{A})$ there will be at most $n_{m-1}=\lfloor(S_m-(r(J_i^{A})+T_i))/T_i\rfloor$  jobs released during $[0,S_m)$. Here $n_{m-1}$ denote the number of jobs released between $J_i^{A}$ and $J_i^{B}$.  When $r(J_i^{{A}})+T_i>S_m$, i.e., $J_i^{A}$ is the only job from $\tau_i$ that generates demand during $[0,\mathit e)$, and  in this case $DEM_i(r_A)$ is equal to $dbf_{MM}(J_i^{A},S_m,\mathit e,m)$. Then we have 
\begin{equation}
\begin{split}
DEM_i(r_A)=\!
 \begin{cases}
&\!\!\!dbf_{MM}(J_i^{A},S_m,\mathit e,m)~\mbox{~~~~~~~~~~~~~~~~~~~~~~~~~~~~~~~~~~~~~~~~~~~~~~~~if~}r(J_i^A)+T_i>S_m\\ 
&\!\!\!dbf_{MM}(J_i^{A},S_m,\mathit e,m)\!+\!n_{m-1}C_i^{m-1}\!+\!dbf_{MM}(J_i^{B},S_m,\mathit e,m)~\mbox{~~otherwise}\\
\end{cases}
\end{split}
\end{equation}
If all tasks have integer release times, we can simply get
\[
dbf_{MM}(\tau_i,S_m,\mathit{e},m)=\max_{r_A\in\{0,-1,-2,\ldots,-T_i\}}\{DEM_i(r_A)\}
\]
For more generic cases, the lemma defines the release pattern when $\tau_i$~$(L_i=m-1)$ generates maximum demand during $[0,S_m)$.


\begin{figure}[tbh]
\centering
\subfigure[$r(J_i^{A})=-D_i^{m-2}+C_i^{m-2}$]{
\label{fig:p3}
        \begin{minipage}[t]{9cm}
        \includegraphics[width=\textwidth]{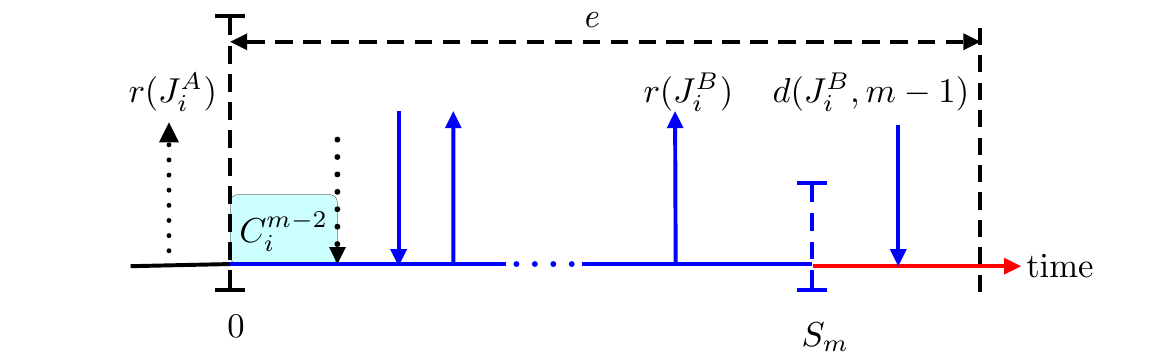}
        \end{minipage}
}
\subfigure[$r(J_i^{A})=mod(\mathit e- D_i^{m-1},T_i)-T_i$]{
\label{fig:p4}
        \begin{minipage}[t]{9cm}
        \includegraphics[width=\textwidth]{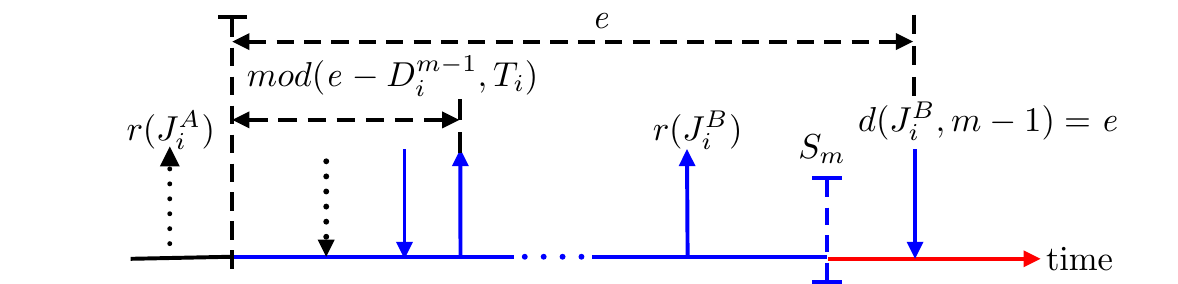}
        \end{minipage}
}
\caption{Job release pattern for tasks with $L_i=m-1$}
\end{figure}

\begin{lemma}
\label{lem:pat2}
Task $\tau_i$ with $L_i=m-1$ generates maximum demand during $[0,S_m)$ if  $r(J_i^{A})=-D_i^{m-2}+C_i^{m-2}$  or $r(J_i^{A})=mod(\mathit e- D_i^{m-1},T_i)-T_i$ (i.e.,  deadline  $d(J_i^{{A}},m-1)=\mathit e\vee d(J_i^{{B}},m-1)=\mathit e$), and all the jobs are released as soon as possible with period $T_i$.  These two pattens are shown in Figure~\ref{fig:p3} and  Figure~\ref{fig:p4}, respectively.

\end{lemma}

\begin{proof}
The proof for Lemma~\ref{lem:pat2} can be found in Appendix~\ref{sec:proof} .
\end{proof}
\noindent Thus using Lemma~\ref{lem:pat2}, we  define $dbf_{MM}(\tau_i,S_m,\mathit e,m)$ for task $\tau_i$ with $L_i=m-1$ as follows.   
\begin{equation}
dbf_{MM}(\tau_i,S_m,\mathit{e},m)=\max\left\{DEM_i\left(-D_i^{m-2}\!+C_i^{m-2}\right),DEM_i\left(mod(\mathit e\!- D_i^{m-1},T_i)\!-T_i\right)\right\}
\end{equation}

\subsubsection {Case when $L_i\geq m$}
The demand generated by task $\tau_i$ with $L_i\geq m$ during the interval $[0,\mathit e)$ is equal to the sum of demand of all jobs released  during $[S_{m-1}=0,\mathit e)$ and the demand of $J_i^{C}$ as shown in Figure~\ref{fig:dbfs1}.

Let $DEM_i(r_C)$ denote the demand that $\tau_i$ generates   during $[0,S_m)$ when $r(J_i^{C})=r_C$ and all jobs are released  as soon as possible with period $T_i$.  Let $n_{m-1}$ denote the number of jobs released during  $[S_{m-1}, r(J_i^{D}))$ and $n_m$ denote the number of jobs released after $S_m$ with deadline $D_i^m$ before $\mathit e$. Given the value of $r(J_i^{C})$, $n_{m-1}=\lfloor  (S_m-r(J_i^{C})-T_i)/T_i \rfloor$.  The first job released after $J_i^{D}$, $J_i^{{D+1}}$, is released at  $r(J_i^{C})+(n_{m-1}+2)\times T_i$, and hence  $n_m=\lfloor (\mathit e-r(J_i^{{D+1}})-D_i^m)/T_i\rfloor+1$. Note that if $n_{m-1}<0$, i.e., $S_m<r(J_i^{C})+T_i$, it implies there does not exist $J_i^{D}$. Therefore we have 
\begin{equation}
\begin{split}
DEM_i(r_C)=
 \begin{cases}
dbf_{MM}(J_i^{C},S_m,\mathit e,m)+n_{m-1}C_i^{m-1}+ dbf_{MM}(J_i^{D},S_m,\mathit e,m)+n_mC_i^{m}\\
\mbox{~~~~~~~~~~~~~~~~~~~~~~~~~~~~~\!~~~~~~~~~~~~~~~~~~~~~~~~~~~~~~~~~~if~}S_m\geq r(J_i^{C})+T_i\\
dbf_{MM}(J_i^{C},S_m,\mathit e,m)+ n_mC_i^{m}~~~~~~~~~~~~~~\mbox{if~~} S_m<r(J_i^{C})+T_i
\end{cases}
\end{split}
\end{equation}
\begin{figure}[h!]
\centering
\includegraphics[width=12cm]{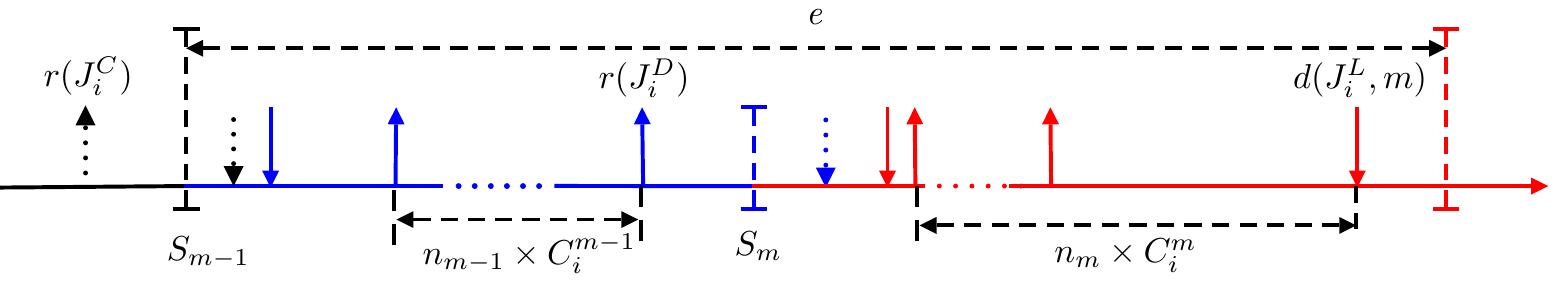}
\caption{Job release pattern for tasks with $L_i\geq m$}
\label{fig:dbfs1}
\end{figure}
If all tasks have integer release times, we can simply get  
\[
dbf_{MM}(\tau_i,S_m,\mathit{e},m)=\max_{r_C\in\{0,-1,-2,\ldots,-T_i\}}\{DEM_i(r_C)\}
\]
For more generic cases, the lemma defines the  release pattern so that $\tau_i$ with $L_i\geq m$ generates maximum demand during $[0,\mathit e)$.


\begin{lemma}
\label{lem:pat3}
Task $\tau_i$ with $L_i\geq m$ generates maximum demand during $[0,\mathit e)$ if $r(J_i^{C})=mod(\mathit e- D_i^{m},T_i))-T_i$ (i.e, the last job released before $\mathit e$ has $d(J_i^L,m)=\mathit e$) or $r(J_i^{C})=mod(\mathit e- D_i^{m-1},T_i)-T_i$ (i.e., $d(J_i^L,m-1)=\mathit e$) or $r(J_i^{C})=-D_i^{m-2}+C_i^{m-2}$ and all the jobs are released as soon as possible with period $T_i$. 
\end{lemma}
\begin{proof}
The proof for Lemma~\ref{lem:pat3} can be found in Appendix~\ref{sec:proof}.
\end{proof}

 Using Lemma~\ref{lem:pat2}, we define the $dbf_{MM}(\tau_i,S_m,\mathit e,m)$ for task $\tau_i$ with $L_i\geq m$ as follows.   
\begin{equation}
\small
\begin{split}
&dbf_{MM}(\tau_i,S_m,\mathit{e},m)=\\
&\max\left\{DEM_i(mod(\mathit e- D_i^{m},T_i)-T_i)),DEM_i(-D_i^{m-2}+C_i^{m-2}), DEM_i(mod(\mathit e- D_i^{m-1},T_i)-T_i)) \right\}
	\end{split}
\end{equation}

\subsection{Demand bound function for task set $\tau$ ($dbf_{MM}(\tau, S_m,\mathit e,m)$)}
Now we have presented the MM dbf for each task.  We can simply add up the demand of all tasks in the system to get the dbf for the  task set $\tau$.
\begin{equation}
dbf_{MM}(\tau, S_m,\mathit e,m)=\sum\limits_{L_i\geq m-1}dbf_{MM}(\tau_i,S_m,\mathit{e},m)
\end{equation}

\text{Discussion:} The MM dbf derived in this section cannot be directly  applied to a more generalized MC model~\cite{EkYi14} where each task has different periods in different criticality modes.  The main challenge in this extension is that we need to  figure out the execution  pattern (i.e., the release time of $J_i^A$ and $J_i^C$)  that can  result in  worst-case demand in the more generalized model.  Of course, we can still compute the demand for all possible release patterns and choose the maximum one, but it can be very computational expensive. For space reason, we would leave this as our future work.

\section{Virtual Deadline Assignment}
\label{sec:T1}
From the previous sections we know  that virtual deadline for each criticality mode plays a key role in shaping the demand of carry-over jobs.  The choice of virtual deadline for each task therefore has a significant impact on the performance of the proposed schedulability test. We can decrease $\tau_i$'s demand in $L_{m}~(m>1)$ criticality mode at the cost of increasing the demand in previous modes. By choosing suitable values for $D_i^{m-1}$ for each $\tau_i$, we can increase the  chances that $\tau$ is schedulable in $L_m$ criticality mode.

The process of finding suitable values for the virtual deadlines is very challenging, because it is infeasible to try all possible combinations of $D_i^m$ for all the tasks in different criticality modes. In this section, we propose a heuristic virtual deadline tuning algorithm which has  pseudo-polynomial time complexity \re{as long as the system utilization of each criticality mode is  bounded by some constant smaller than 1}.

In Section~\ref{sec:ExistDBF} and Section~\ref{sec:GIT} we introuce the existing single-mode dbf-based schedulability test and our proposed multi-mode dbf-based test, respectively. A MC task system $\tau$ is schedulable in $L_m$ mode assuming no deadline miss happens in previous criticality modes if one of the following two condition holds.
\begin{enumerate}
	\item $\mathbf{CN^S_m}$: $\forall \mathit{e}\in\{1,2,\ldots,\mathit e^{max}\}:~dbf_{SM}(\tau,\mathit{e},m)\leq \mathit e $~(Theorem~\ref{theorem:theorem0})~~. \\
	\item $\mathbf{CN^M_m}$:  $\forall \mathit{e}\in\{1,2,\ldots,\mathit e^{max}\}:$ $\forall S_m\in\{1,2,\ldots,\mathit e\}:~dbf_{MM}(\tau,S_m,\mathit{e},m)\leq \mathit e $~(Theorem~\ref{theorem:MM}).
\end{enumerate}


Though the  complexity of both of the above  conditions is pseudo-polynomial, there is a quadratic increase in complexity to check condition  $\mathbf{CN^M_m}$. \re{For a certain $\mathit e^{max}$, while the SM dbf test~\cite{EkYi14} need to calculate the demand for  $\mathit e^{max}$ times, the MM dbf test  need to calculate the demand for $\mathit e^{max}\times \mathit e^{max}$ times.} Therefore, as shown in line 18-27 in Algorithm~\ref{alg:tunemode},   we  use the condition  $\mathbf{CN^M_m}$ as  a complement, i.e., only when condition  $\mathbf{CN^S_m}$ fails at  a certain $\mathit e$, we will check whether condition  $\mathbf{CN^M_m}$ is satisfied for that time interval length.  Thus, once we find an $\mathit e_{f}$ such that $dbf_{SM}(\tau,\mathit{e}_f,m)>\mathit e_f$, we use the Multi-Mode demand bound function to check whether the following holds.
\begin{align*}
\forall~\mathit{e}\in\{\mathit e_f,\mathit e_f+1,\ldots,\mathit e^{max}\}:~dbf_{MM}(\tau,S_m=\mathit e-\mathit e_{f},\mathit e,m)\leq \mathit e
\end{align*}
If the above inequalities hold, then it means no deadline miss happens $e_f$ time units after $S_m$. Otherwise  a candidate task will be chosen and its virtual deadline $D_i^{m-1}$ is reduced by $1$.

In Algorithm TuneMode(m), if both $\mathbf{CN^M_m}$ and $\mathbf{CN^S_m}$ fail at a certain $\mathit e$, then a candidate task is chosen and its virtual deadline $D_i^{m-1}$ is reduced by one unit. By doing this, we can reduce the demand of these tasks in $L_{m}$ criticality mode. It does this without considering  the schedulability in $L_{m-1}$ criticality mode itself, hoping that $\tau$ in  $L_{m-1}$  criticality mode can later be made schedulable by decreasing deadlines for $L_{m-2}$ criticality mode, i.e., $D_i^{m-2}$. Once $D_i^{m-1}$ decreases to $C_i^{m-1}$, $\tau_i$ is eliminated from the candidate sets $\Psi_m$, and then we need to find another candidate task. We repeat the above steps  until  TuneMode(m)  returns true or the candidate task set $\Psi_m$ becomes empty. 

\begin{algorithm}[ht]
\SetAlgoNoLine
\SetCommentSty{small}
$\Psi_m\leftarrow \{\tau_i | L_i\geq m\}$\;
changed$\leftarrow$\textbf{False}\;
\While{$changed=\mathbf{True}$}
{
	changed$\leftarrow$\textbf{False}\;
	\For{$\mathit e_1\in \{1,2,\ldots,\mathit e^{max}\}$}
	{
		\If{$m=2$ and $dbf_{ISM}(\tau,\mathit e_1,1)>\mathit{e}_1$ }
		{
		\If{$changed=\mathbf{True}$} 
		{

			$D_i^{1}\leftarrow D_i^{1}+1$\;
			$\Psi_m$.remove($\tau_i$)\;
			changed$\leftarrow$\textbf{True}\;
			break\;

		}
		\Else
		{
			return \textbf{False}
		}
		}
		\If{$dbf_{SM}(\tau,\mathit e_1,m)>\mathit{e}_1$}
		{
			$\mathbf{CN^M_m}\leftarrow\mathbf{True}$\;
			\For{$\mathit e_2\in \{\mathit e_1,\mathit e_1+1,\ldots,\mathit e^{max}\}$}
			{
				\If{$dbf_{MM}(\tau,S_m=\mathit e_2-\mathit e_1,\mathit e_2)>\mathit e_2$}
				{
					$\mathbf{CN^M_m}\leftarrow\mathbf{False}$\;
					break\;
				}

			}
		\If{$\mathbf{CN^M_m}=\mathbf{True}$}
		{
				continue\;
		}

		\If{$\Psi_m=\emptyset$}
		{
			return \textbf{False}
		}
		$\tau_i\leftarrow$ find\_candidate\_m($\Psi_m$, $\mathit e)$~(Algorithm~\ref{alg:findcandidate})\;
		$D_i^{m-1}\leftarrow D_i^{m-1}-1$\;
		\If{$D_i^{m-1}< C_i^{m-1}$ }
		{
			$D_i^{m-1}\leftarrow D_i^{m-1}+1$\;
		  $\Psi_m$.remove($\tau_i$)\;
		}
		\Else{
			changed$\leftarrow$\textbf{True}\;
			break\;
		}

	}

	}

}
\Return \textbf{True}
\caption{TuneMode(m)}
\label{alg:tunemode}
\end{algorithm}

An exceptional scenario is that: if we decrease $D_i^1$ to reduce the demand in $L_2$ criticality mode, $\tau$ may become unschedulable in $L_1$ criticality mode. However there does not exist $D_i^0$, and hence we cannot tune $D_i^0$ to reduce the demand in  $L_1$ criticality mode. Therefore in this case, TuneMode(m) will undo the changes to deadlines  that would make $\tau$ unschedulable in $L_1$ criticality mode. \re{Here undo the changes  to deadlines means $D_i^{1}=D_i^{1}+1$, and then the candidate task $\tau_i$ is removed from  $\Psi_m$.}

Algorithm~\ref{alg:tunesystem} (TuneSystem($\tau$))   applies   Algorithm~\ref{alg:tunemode}  (TuneMode(m)) on all the criticality modes starting from $L_{M}$ criticality mode and proceeding in a reverse order, until it has successfully tuned the deadlines in all the criticality modes or failed to do so in  some criticality mode. \re{Therefore, the complexity of Algorithm~\ref{alg:tunesystem} (TuneSystem($\tau$)) increases linearly as the number of levels increases because the complexity of Algorithm~\ref{alg:tunemode}  (TuneMode(m)) is independent of the value of $m$.}
\begin{algorithm}[t]
\SetCommentSty{small}
\SetAlgoNoLine
$M\leftarrow\max_{\tau_i\in \tau}\{L_i\}$\;
\For {$m\in \{M~\mbox{to}~2\}$} 
{
	\If {$TuneMode(m)=\mathbf{False}$}
	{
		\Return $\mathbf{False}$
	}
}
\Return $\mathbf{True}$
\caption{TuneSystem($\tau$)}
\label{alg:tunesystem}
\end{algorithm}

Finally, to select a candidate task for deadline tuning in each iteration, TuneMode(m) uses Algorithm~\ref{alg:findcandidate} (find\_candidate\_m($\Psi$, $\mathit e$)). Note that,  GreedyTuning~\cite{EkYi12} uses a very simple metric to choose a candidate task, i.e., $\tau_i$ with  the maximum
\[\Delta_i=dbf_{SM}(\tau_i,\mathit e,m)-dbf_{SM}(\tau_i,\mathit e-1,m)\]
is always chosen as a candidate task. Here  $\Delta_i$ denote the demand change of $\tau_i$ using Single-Mode dbf if $D_i^{m-1}\leftarrow D_i^{m-1}-1$.

However, there are many other parameters which are also important in choosing a good candidate task. Therefore, in Algorithm~\ref{alg:findcandidate},  we  extend this metric to additionally consider other factors in candidate selection such as the impact of change in virtual deadline on the schedulability for previous mode ($L_{m-1}$). In Section~\ref{sec:eva}, we  show that the new metric for choosing a candidate task outperforms the one in~\cite{EkYi14}.

To maximize schedulability in $L_m$ mode by reducing the demand of a candidate task, a task with larger $\Delta_i$  is preferable. Meanwhile when a candidate  task $\tau_i$'s virtual deadline $D_i^{m-1}$ decreases by one, the impact on the demand of $L_{m-1}$ mode is different.

Suppose task set $\tau$  has  task $\tau_1$ and $\tau_2$ where  $T_1=D_1^2=D_1^1=10000$ and $\tau_2$ has $T_2=D_2^2=D_2^1=2$. Also, suppose  $L_2$ mode of $\tau$ is currently not schedulable but can be tuned to become schedulable if either $D_1^1$ or  $D_2^1$ decreases by one.  The impact of  $D_1^1\leftarrow D_1^1-1$ or $D_2^1\leftarrow D_2^1-1$ is different, i.e.,  $ \frac{\Delta D_1^1}{D_1^1}=0.0001$ and $\frac{\Delta D_2^1}{D_2^1}=0.5$.  If $D_1^1\leftarrow D_1^1-1$, the system demand of $L_1$ mode stays the same for any time interval length $\mathit e(<9999)$.  As a result,  the $L_{1}$ mode is easier  to be schedulable  if $D_1^1$ decreases to $9999$ compared to the case if $D_2^1$ decreases to $1$. From the above discussion, we know a task with larger $\Delta_i$ and $D_i^{m-1}$ is more likely to become a better candidate task. Therefore, among all possible candidate tasks in $\Psi_m$, we use $\Delta_i\times D_i^{m-1}$ as the main metric to choose a candidate task.

Until now, we did not consider case when all tasks have $\Delta_i=0\Rightarrow \Delta_i\times D_i^{m-1}=0$. This means that,  there does not exist a candidate task so that the total demand would decrease if $D_i^{m-1}\leftarrow D_i^{m-1}-1$. As shown in Figure~\ref{fig:dbfc3}, if $\Delta_i=0$,  it must be that  $len_i=mod(\mathit e,T_i)-(D_i^{m}-D_i^{m-1})-C_i^{m-1}>0$. The demand of $\tau_i$ would start to decrease if $D_i^{m-1}\leftarrow D_i^{m-1}-len_i-1$. Therefore, we also choose those tasks with the $\min\{\max\{0,len_i\}\}$ among those candidate tasks. Hence in Algorithm~\ref{alg:findcandidate}, we sort tasks in $\Psi_m$ with first key $\Delta_i\times D_i^{m-1}$ in ascending order and then with second key $\max\{0,len_i\}$ in descending order. 

\begin{algorithm}[h]
\caption{find\_candidate\_m($\Psi$, $\mathit e$)}
\label{alg:findcandidate}
\SetCommentSty{small}
\SetAlgoNoLine
\CommentSty{\color{blue}}
$Q\leftarrow \mbox{empty queue}$\;
\For{$\tau_i  \in \Psi$}
{
	$Q.insert(\Delta_i\times D_i^{m-1},\max\{0,len_i\},\tau_i)$\;

}
$Q.sort(\mbox{ascending},\mbox{descending})$ \tcp*[l]{\textcolor{blue}{sort tasks in $\Psi$ with first key $\Delta_i\times D_i^{m-1}$ in ascending order and then with second key $\max\{0,len_i\}$ in descending order}}
 $\tau_i \leftarrow Q.\mbox{pop()}$\tcp*[l]{\textcolor{blue}{return the corresponding task}}
\Return $\tau_i$
\end{algorithm}

\section{Evaluation}
\label{sec:eva}
In this section we evaluate the ability of the proposed virtual deadline assignment strategy and the MM dbf based test (i.e, Algorithm~\ref{alg:tunesystem}) to schedule MC task systems.  We use acceptance ratios, i.e., the fraction of schedulable task sets, as the metric to evaluate our proposed approach. 

We aim to compare with the existing work named GreedyTuning~\cite{EkYi14} based on SM dbf. From extensive experiments,  it has already been shown in ~\cite{EkYi14} that this work outperforms  existing studies (e.g., \cite{BBA12,BBG11,Ves07}) for a variety of generic real-time systems.  Therefore, through comparison with this SM dbf work, we aim to show in this section that the techniques proposed in this paper also outperform those studies, including the one based on SM dbf test. 

We consider MC sporadic tasks scheduled on a dedicated unit speed uniprocessor platform. We will study the impact of varying parameters of tasks on the acceptance ratios of these approaches: 1):~GreedyTuning (GT)~\cite{EkYi14}, 2):~SM dbf test from~\cite{EkYi14} with the improved deadline assignment strategy (i.e,~Algorithm~\ref{alg:findcandidate}) (GTI),  3):~our improved dbf-based test~(IMPT, i.e.,~Algorithm~\ref{alg:tunesystem}).


\subsection{Task Set Generation}
Suppose task set $\tau $ is a empty task set ($\tau=\emptyset$) initially. Randomly generated tasks are  added to the task set $\tau$ repeatedly until certain requirements are met. The parameters of each task is controlled by the following parameters.
\begin{itemize}
\label{set:1}
	\item $P(m)$ denotes the probability that $\tau_i$ has $L_i=m$.
	\item $C_i^1$ is drawn using an uniform distribution over $[1,10]$.
	\item $RC_m$ denotes the maximum ratio of $C_i^m/C_i^{m-1}$.
	\item $C_i^m$ is  drawn using an uniform distribution over $[C_i^{m-1},RC_m\times C_i^{m-1}]$.
	\item $T_i$ is drawn using an uniform distribution over $[C_i^{Li},200]$.
  	\item $D_i$ is drawn using an uniform distribution over $[D_i^{MIN},T_i]$  where $D_i^{MIN}=\lfloor C_i^{Li}+RD\times(T_i-C_i^{L_i}) \rfloor$, and $RD\in[0,1]$
\end{itemize}

Let $U_{\tau}=\max\limits_{m\in\{1,2,\ldots, M\}}\{\sum\limits_{L_i\geq m}\frac{C_i^m}{T_i}  \}$  denote the  utilization bound of a task set $\tau$.  For a given utilization bound, our generation procedure requires  $U_{\tau}$ to fall within the small interval between $[Ubound-\epsilon, Ubound]$ ($\epsilon=0.005$).  As long as $U_{\tau}<Ubound-\epsilon$,  a new task  will be  randomly generated and  added to $\tau$. Once $U_{\tau}$ of $\tau$ falls within the range $[Ubound-\epsilon, Ubound]$, the generation procedure for $\tau$ is considered complete.  However if $U_{\tau}$ becomes greater than $Ubound$ after a new task  $\tau_i$ is added to $\tau$, we discard the whole task set and start with a new empty task set.

\begin{figure}[h!]
\centering
\includegraphics[scale=0.37]{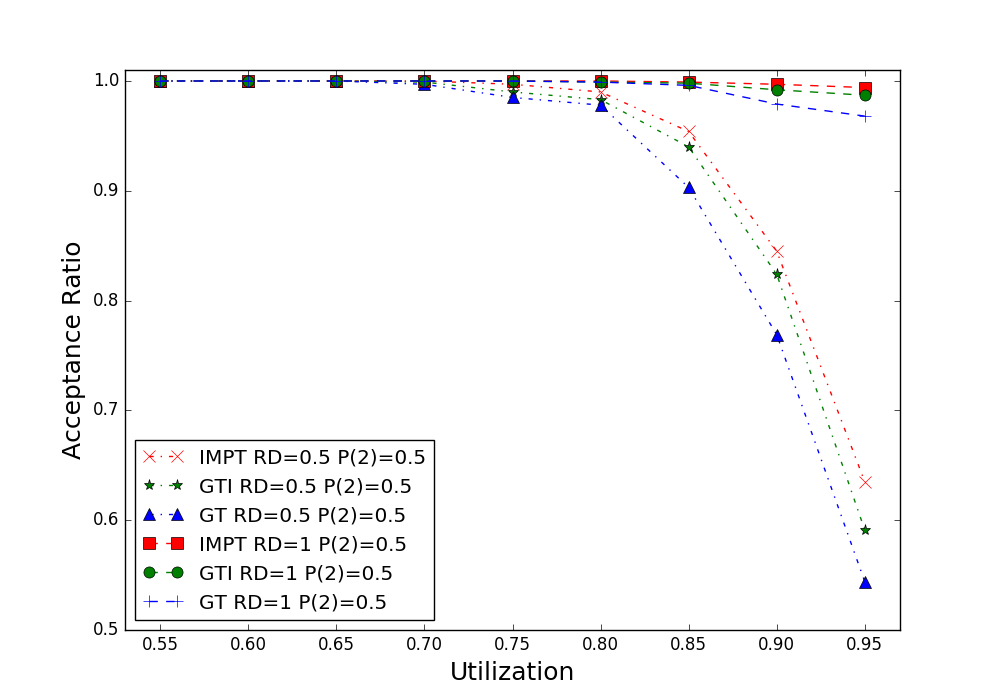}
\caption{Dual-Criticality Task Systems with $RD\in\{0.5,1\}$ and $P(2)=0.5$}
\label{fig:L2}
\end{figure}
\vspace{-20mm}

\begin{figure}[h!]
\centering
\includegraphics[scale=0.37]{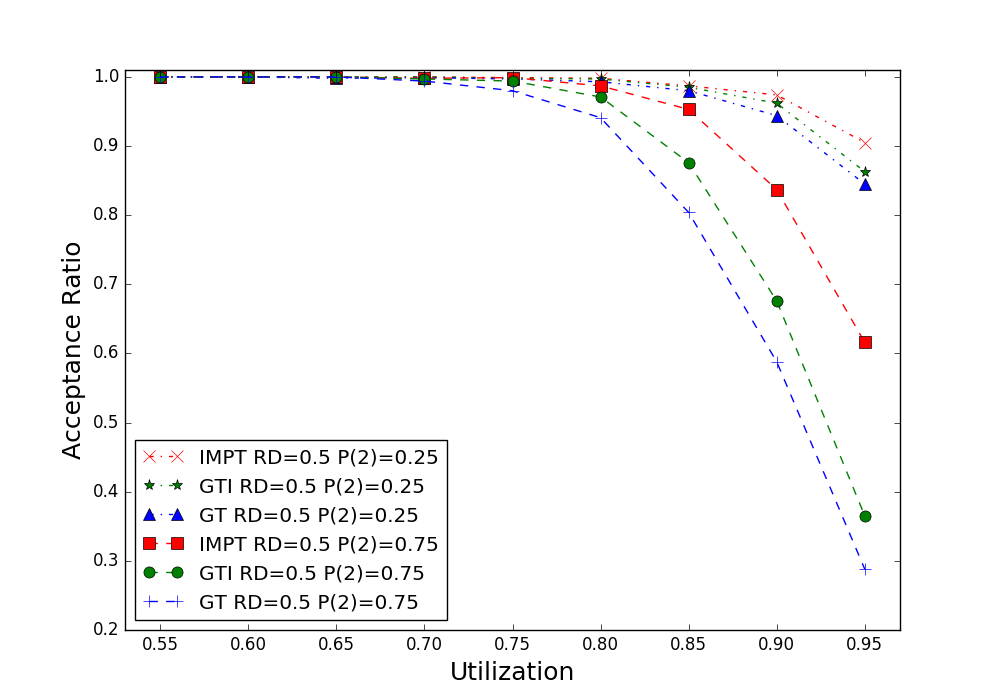}
\caption{Dual-Criticality Task Systems with $RD=0.5$ and $P(2)\in \{0.25,0.75\}$}
\label{fig:L2s}
\end{figure}





\subsection{Evaluation of Dual-Criticality Systems}

Figure~\ref{fig:L2} shows the acceptance ratio  as a function of  utilization  ($Ubound$) of task sets generated with $RC_2=3$, $P(1)=P(2)=0.5$ and $RD\in\{0.5,1\}$.   Each data is based on 1000 randomly generated task sets.   As shown in Figure~\ref{fig:L2}, our improved dbf-based schedulability test (IMPT)  strictly outperforms  GreedyTuning (GT). However the performance gap between  GT and IMPT is very small when $RD=1$, which implies GT already does quite well in scheduling dual-criticality implicit deadline task systems.


Figure~\ref{fig:L2s}  shows the acceptance ratios of task sets  generated with  $RD=0.5$, $RC_2=3$ and  $P(2)\in\{0.25,0.75\}$.  The acceptance ratios of task sets with $RD=0.5$, $RC_2=3$ and   $P(2)=0.5$ can be found in Figure~\ref{fig:L2}.  We can observe that as $P(2)$ increases from $0.25$ to $0.75$, the acceptance ratio of GT  drops quickly.  Even though the SM dbf test with the improved deadline assignment strategy (GTI) always has a higher  acceptance ratio than GT, its  acceptance ratio is much lower than IMPT when $P(2)=0.75$. One interpretation for this trend is that SM dbf test is not good at scheduling systems with a larger percentage of high critical tasks. On the other hand,  the acceptance ratio of  IMPT drops much slower,  and its acceptance ratio becomes almost two times as much as GT when $P(2)=0.75$ and $Ubound=0.95$. In fact,  the  acceptance ratio of IMPT when $P(2)=0.75$ is closer to its  acceptance ratio when $P(2)=0.5$.

\subsection{Evaluation of Multi-Level Systems}
In this section we compare the acceptance ratios of different tests for multi-criticality task systems. Figure~\ref{fig:L3I}  shows the acceptance ratios  of three level  systems with $P(1)=P(2)=P(3)=1/3$, $RC_2=RC_3=2$ and $RD\in\{0.5,1\}$.  As we can observe,  the performance of all the three approaches drops, but  IMPT very well outperforms GT  even for small utilization bounds.

\begin{figure}[h!]
\centering
\includegraphics[scale=0.37]{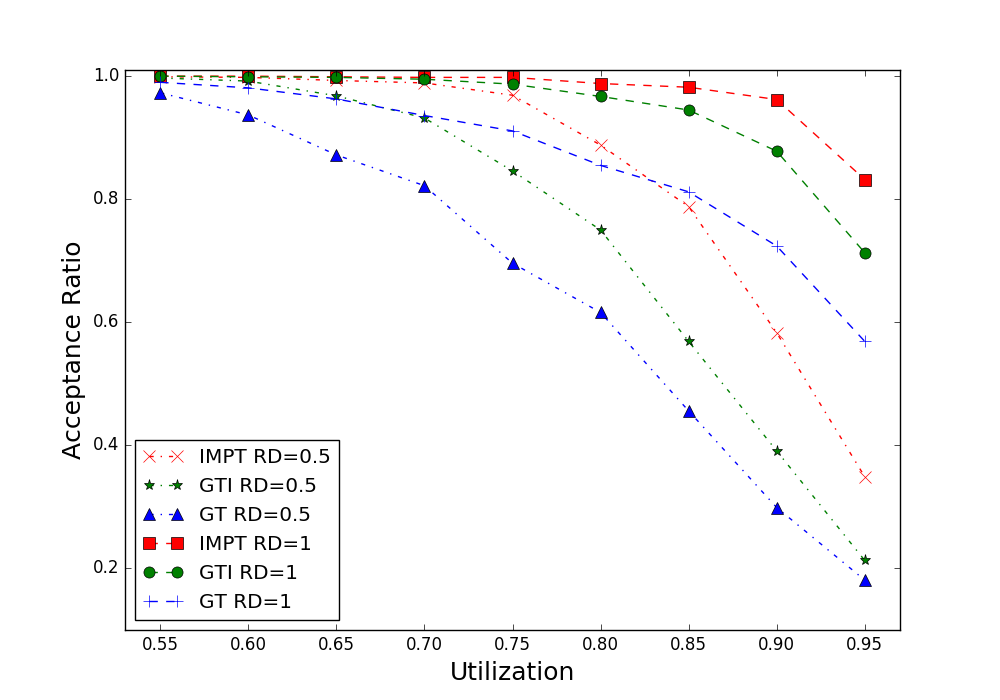}
\caption{Three Level Task Systems}
\label{fig:L3I}
\end{figure}


To study how the acceptance ratio changes as the number of criticality levels increases, we here present the weighted acceptance ratios of different approaches for each  criticality level.  Suppose $A(U)$ is the  acceptance ratio of a certain approach for  utilization $U\in[0.55,0.6,0.65,\ldots,0.95]$, then the weighted acceptance ratio is defined as
\[
\frac{\sum_{U\in\{0.55,0.6,0.65,\ldots,0.95\}}A(U)\times U }{\sum_{U\in\{0.55,0.6,0.65,\ldots,0.95\}}U}.
\]

\begin{figure}[tbh]
\centering
\subfigure[Weighted Acceptance Ratio $RD=1$]{
        \begin{minipage}[t]{0.48\textwidth}
        \includegraphics[width=\textwidth]{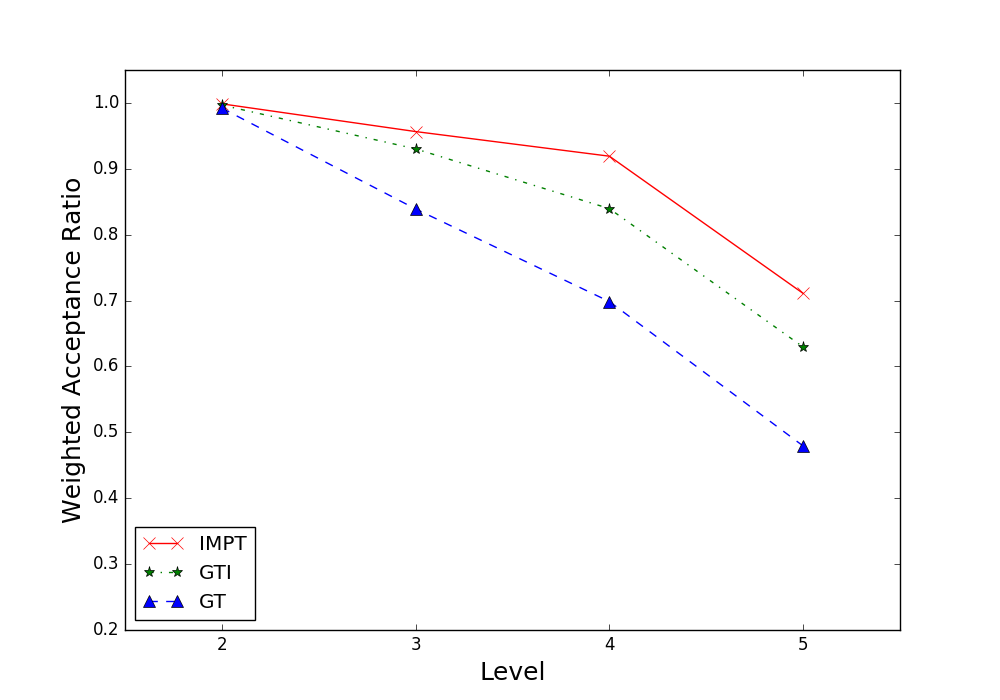}
        \end{minipage}
}
\subfigure[Weighted Acceptance Ratio $RD=0.5$]{
        \begin{minipage}[t]{0.48\textwidth}
        \includegraphics[width=\textwidth]{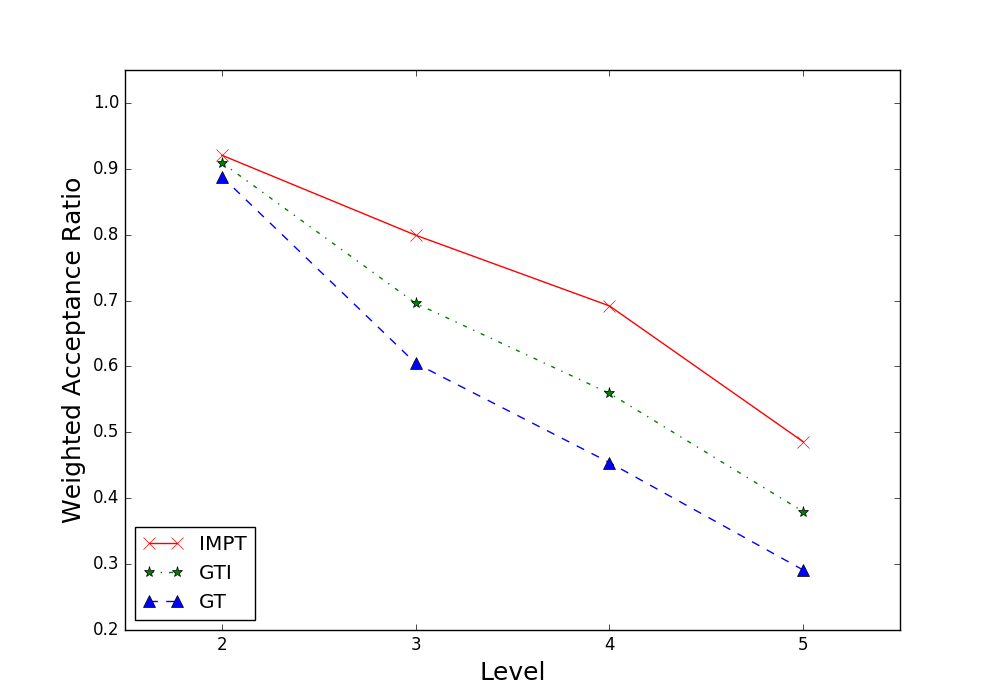}
        \end{minipage}
}
\caption{Weighted Acceptance Ratio}
\label{fig:weight}
\end{figure}




As we can observe,  the gap in weighted acceptance ratio  between  GT  and IMPT  is in fact quite small when the criticality level is equal to $2$.  However, this gap become much larger  when the criticality level becomes greater than $2$. The acceptance ratio of IMPT  also becomes lower (especially for the case when $RD=0.5$)  as the number of criticality level increases. One interpretation for this trend may be that, there is not enough space to tune the virtual deadlines since we need to set multiple virtual deadlines for each task.  However, note that,  it does not mean that the proposed approach IMPT has very poor performance, and can schedule only a small portion of generated task sets. Since there is no known exact feasibility test for MC task systems, we are unable to eliminate all the infeasible from our experiments.  Nevertheless, we can  still conclude that the proposed approach in this paper outperforms GT for a variety of generic systems.



\section{Conclusions}
\label{sec:conclusions}

We first introduced the existing single-mode demand bound functions~\cite{EkYi14}, which characterize the demand of mixed-criticality sporadic tasks. They use a pessimistic  upper bound to characterize demand of carry-over jobs by assuming that the previous criticality mode is schedulable.  As a result, the single-mode dbfs  over estimate the demand of carry-over jobs. Due to the drawback of the single-mode dbf based test, it has a severe problem  that its performance decreases significantly as the number of criticality levels increases.

To avoid the problem of single-mode dbf based schedulability test, we propose  multi-mode dbfs which consider the execution demand in the previous criticality mode to determine the remaining execution for carry-over jobs.  The proposed multi-mode dbf based test can avoid the problem of the single-mode dbf at the cost of a quadratic increase in the complexity. In practice it could be computationally expensive if we use the multi-mode dbf based test directly. Therefore we propose a novel heuristic deadline tuning algorithm which uses the multi-mode dbf as a complement to reduce the off-line computation time.  Finally we show that  our proposed approach  outperforms single-mode dbf based schedulability test~\cite{EkYi14} from experimentation. 

Often, EDF is quoted as being too unpredictable in case of overloads since it is practically impossible to predict which jobs will suffer the extra delays. This is not the case for mixed-criticality systems, because more important (or critical) tasks will be prioritized in an overload situation.

Though we use the multi-mode dbfs as a complement to reduce off-line computation time, it still takes a lot of time compared to single-mode dbf based test. As future work we plan to find a strategy to reduce the computation time of our proposed approach. One limitation of  the multi-mode dbf is that it is limited to constrained deadline MC systems. Even though it seems to be straightforward to extend it to arbitrary deadline, it can become very pessimistic because there would exist more than one carry-over job.  Therefore in the future we also plan to address this problem and extend the multi-mode dbf to arbitrary deadline mixed-criticality systems.

\appendix
\section{APPENDIX}
\subsection{Demand bound function for carry-over jobs}
\label{sec:dbfjob}
In Section~\ref{sec:GIT} we have derived the MM dbf for individual tasks as well as for the entire task system. These functions use the dbf for the carry-over jobs at mode switches $S_m-1$ and $S_m$ (i.e., jobs $J_i^A$, $J_i^B$,$J_i^C$ and $J_i^D$ defined in Section~\ref{sec:GIT}). For simplicity we assume $S_{m-1}=0$ because the dbf for carry-over jobs is independent of $S_{m-1}$.  In this section, we derive an bound of the demand of  these jobs.  In order to present the dbf, we will make use of the following conditions on the virtual deadlines of these jobs in various modes.
\begin{description}
	\item[Condition (1)] $d(J_i^{X},m-2)=r(J_i^X)+D_i^{m-2}<S_{m-1}$. 
	\item[Condition (2)] $d(J_i^{X},m-1)=r(J_i^X)+D_i^{m-1}<S_m$.
	\item[Condition (3)] $d(J_i^{X},m)=r(J_i^X)+D_i^{m}\leq \mathit e$.
	\item[Condition (4)] $d(J_i^{X},m-1)=r(J_i^X)+D_i^{m-1}\leq \mathit e$.
\end{description}
\vspace{-5mm}
\begin{figure}[htb]
\centering
\includegraphics[width=0.65\textwidth]{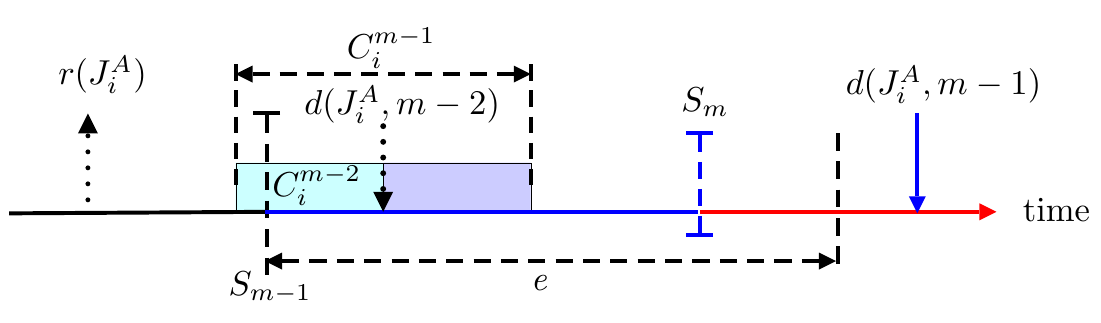} 
\caption{$J_i^{A}:\neg (1)\wedge\neg(4)$}
\label{fig:joba}
\end{figure}
$J_i^{A}$ would either only experience mode-switch at $S_{m-1}$ or mode-switches at $S_m$ and $S_{m-1}$. If $J_i^{A}$ has its virtual deadline $d(J_i^{A},m-2)<S_{m-1}$ (i.e., condition (1)),  then it must already finish before $S_{m-1}$. On the other hand if $\neg(1):d(J_i^{A},m-2)\geq S_{m-1}$ and $(4):d(J_i^{A},m-1)\leq \mathit e$, the execution demand of $J_i^{A}$ is bounded by $\min\{d(J_i^{A},m-2)-S_{m-1},C_i^{m-2}\}+C_i^{m-1}-C_i^{m-2}$. The  demand of $J_i^{A}$ after $S_{m-1}$ is also bounded by $S_m$ because $J_i^{A}$ would not execute after $S_m$ . An extreme case is when $\neg(4):r(J_i^{A})+D_i^{m-1}>\mathit e$  as shown in Figure~\ref{fig:joba}, and in this case $J_i^{A}$ generates 0 demand after $S_{m-1}$ because its deadline is out of the time interval of interest. Hence we generalize $dbf_{MM}(J_i^{A},S_m,\mathit e,m)$ as follows.
\begin{equation}
\small
\label{eqn:joba}
\begin{split}
&~~dbf_{MM}(J_i^{A},S_m,\mathit e,m)=\\
&
\begin{cases}
\min\left\{\min\left\{r(J_i^{A})+ D_i^{m-2},C_i^{m-2}\right\}+C_i^{m-1}-C_i^{m-2},S_m\right\}&\mbox{If }\neg(1)\wedge (4)\\
  0 &\mbox{Otherwise}
\end{cases}
\end{split}
\end{equation}

$J_i^{B}$ is similar to $J_i^{A}$ except that we can ignore condition (1) because it can never be true. To maximize $J_i^{B}$'s  demand, we assume it executes continuously from $r(J_i^{B})$. Again if its deadline $d(J_i^{B},m-1)>\mathit e$, it would generate 0 demand. Hence, we have
\begin{equation}
\label{eqn:jobb}
\begin{split}
dbf_{MM}(J_i^{B},S_m,\mathit e,m)=
   \begin{cases}
\min\left\{C_i^{m-1},S_m-r(J_i^{B})\right\}&\mbox{If } (4)\\
0&\mbox{Otherwise}
    \end{cases}
\end{split}
\end{equation}

Similar to $J_i^{A}$, if $(1):d(J_i^{C},m-2)<S_{m-1}=0\vee \neg(4):d(J_i^{C},m-1)>\mathit e$ , $J_i^{C}$ generates 0 demand during $[0,\mathit e)$.  If $\neg(1)\wedge (2)$, $J_i^{C}$ would already finish by $S_m$, and hence would generate demand equal to $C_i^{m-1}-C_i^{m-2}+\min\{d(J_i^{C},m-2)-S_{m-1},C_i^{m-2}\}$.  If  $\neg(1)\wedge \neg(2)\wedge (3)$, $J_i^{C}$ would generate $C_i^{m}-C_i^{m-2}+\min\{d(J_i^{C},m-2)-S_{m-1},C_i^{m-2}\}$ demand as shown in Figure~\ref{fig:jobc}. If  $\neg(1)\wedge \neg(2)\wedge\neg(3)\wedge (4)$, $J_i^{C}$ would not generate demand after $S_m$ because its deadline $d(J_i^{C},m)$ is out of the interval of interest. In this case its demand is bounded by $S_m$. Hence, we have
\begin{figure}[thb]
\centering
\includegraphics[scale=0.85]{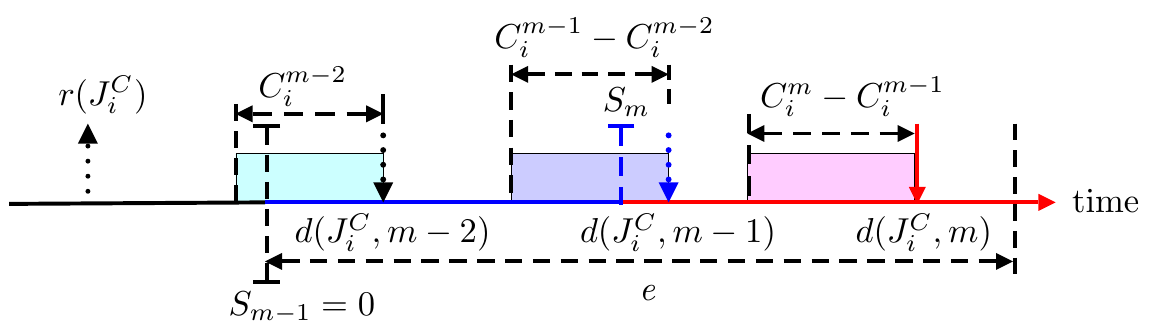} 
\caption{$J_i^{C}:\neg(1)\wedge\neg(2)\wedge(3)$}
\label{fig:jobc}
\end{figure}

\begin{equation}
\small
\label{eqn:dbffi}
\begin{split}
&~~dbf_{MM}(J_i^{C},S_m,\mathit e,m)=\\&
   \begin{cases}
	 0& \mbox{If } (1)\vee\neg(4)\\
	 C_i^{m-1}-C_i^{m-2}+\min\left\{d(J_i^{C},m-2),C_i^{m-2}\right\}&\mbox{If }
	 \neg(1)\wedge (2)\\
	 C_i^{m}-C_i^{m-2}+\min\left\{d(J_i^{C},m-2),C_i^{m-2}\right\}&\mbox{If }
	 \neg(1)\wedge \neg(2)\\&\wedge (3)\\
	\min\left\{S_m,C_i^{m-1}-C_i^{m-2}+\min\left\{d(J_i^{C},m-2),C_i^{m-2}\right\}\right\} &\mbox{If} \neg(1)\wedge \neg(2)\\&\wedge \neg(3) \wedge (4)
   \end{cases}
\end{split}
\end{equation}



$J_i^{D}$ behaves similar to $J_i^{C}$ except we can ignore condition (1) because it can never be true. Hence, we have
\begin{equation}
\small
\label{eqn:dbffi}
\begin{split}
dbf_{MM}(J_i^{D},S_m,\mathit e,m)=
   \begin{cases}
	 0& \mbox{If }\neg (4)\\
	 C_i^{m-1}&\mbox{If } (2)\\
	 C_i^{m}&\mbox{If }
	 \neg(2)\wedge (3)\\
	\min\{C_i^{m-1},S_m-r(J_i^{D})\}  &\mbox{If}  \neg(2)\wedge \neg(3) \wedge (4)
   \end{cases}
\end{split}
\end{equation}

\subsection{Proofs}
\label{sec:proof}
\ae{
\begin{proof}[for Lemma~\ref{lem:pat2}]
\textbf{(C1):} If $\mathit e\leq D_i^{m-1}$, obviously the demand of $\tau_i$ maximizes when  $r(J_i^{A})=mod(\mathit e- D_i^{m-1},T_i)-T_i$,~i.e., $d(J_i^{A},m-1)=\mathit e$ as shown in Figure~\ref{fig:c1}.
\begin{figure}[h!]
\centering
\includegraphics[scale=0.75]{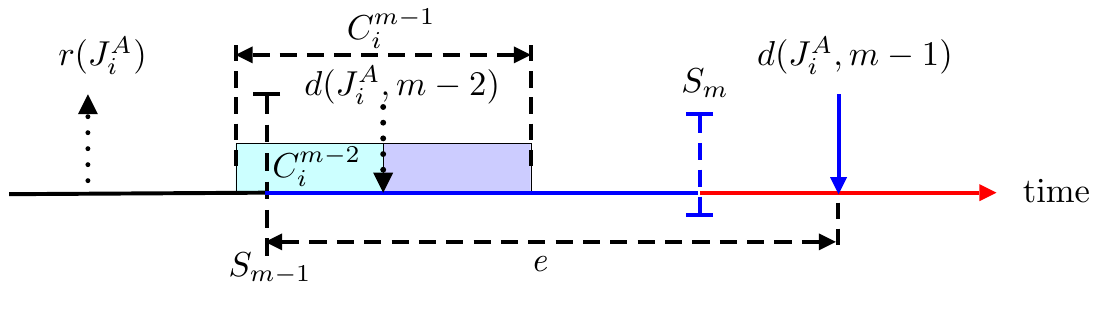} 
\caption{Case when  $\mathit e\leq D_i^{m-1}$}
\label{fig:c1}
\end{figure}

\textbf{(C2):} If $(S_m\leq C_i^{m-1}\wedge \mathit e> D_i^{m-1})$, we can find the demand of $\tau_i$ maximizes when $r(J_i^{A})=-D_i^{m-2}+C_i^{m-2}$, which is is bounded by $S_m$ as shown in Figure~\ref{fig:c2}.
\begin{figure}[h!]
\centering
\includegraphics[scale=0.7]{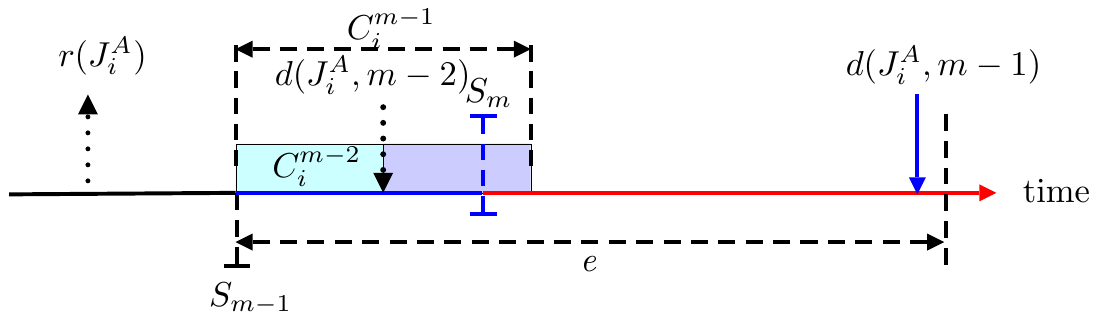} 
\caption{Case when  $(S_m\leq C_i^{m-1}\wedge \mathit e> D_i^{m-1})$}
\label{fig:c2}
\end{figure}
Now we can exclude the above two cases, and we only need to consider the case when $\mathit e> D_i^{m-1} \wedge S_m> C_i^{m-1}$.

\textbf{(C3):} If  $r(J_i^{B})+D_i^{m-1} \leq \mathit e$ when $r(J_i^{A})=-D_i^{m-2}+C_i^{m-2}$ as shown in Figure~\ref{fig:c3}.
\begin{figure}[h!]
\centering
\includegraphics[scale=0.7]{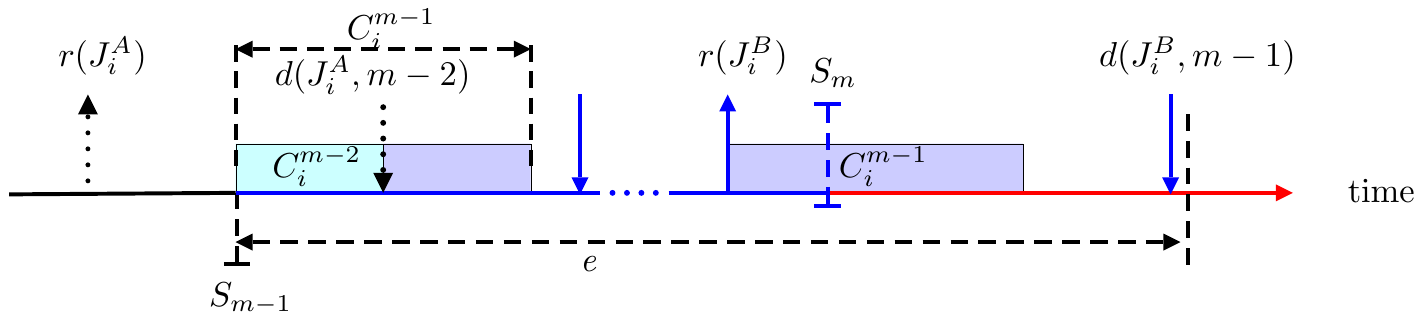} 
\caption{Case C3}
\label{fig:c3}
\end{figure}
As we shift the release pattern left, i.e., $r(J_i^{A})=-D_i^{m-2}+C_i^{m-2}+x|x<0$, the demand change of $J_i^A$, $\Delta_1(x)$, is of the following form.
\[
\begin{split}
 	\Delta_1(x)=
   \begin{cases}
	x &\mbox{ If } x \in [-C_i^{m-2},0)\\
	-C_i^{m-1} & \mbox{ If } x \in [-C_i^{m-2}+D_i^{m-2}-T_i, -C_i^{m-2})
   \end{cases}
\end{split}
\]
Meanwhile the demand of $J_i^B$ will at most increase  linearly, and hence the total demand of $\tau_i$ will stay the same or decrease as we shift the pattern left. As we shift the release pattern right,  the demand of $\tau_i$ will only decrease or stay the same. Therefore in this case the demand of $\tau_i$  maximizes when $r(J_i^{A})=-D_i^{m-2}+C_i^{m-2}$. 

\textbf{(C4):}  If  $r(J_i^{B})+D_i^{m-1} > \mathit e$ when $r(J_i^{A})= -D_i^{m-2}+C_i^{m-2} $ as shown in Figure~\ref{fig:c4}.
\begin{figure}[h!]
\centering
\includegraphics[scale=0.7]{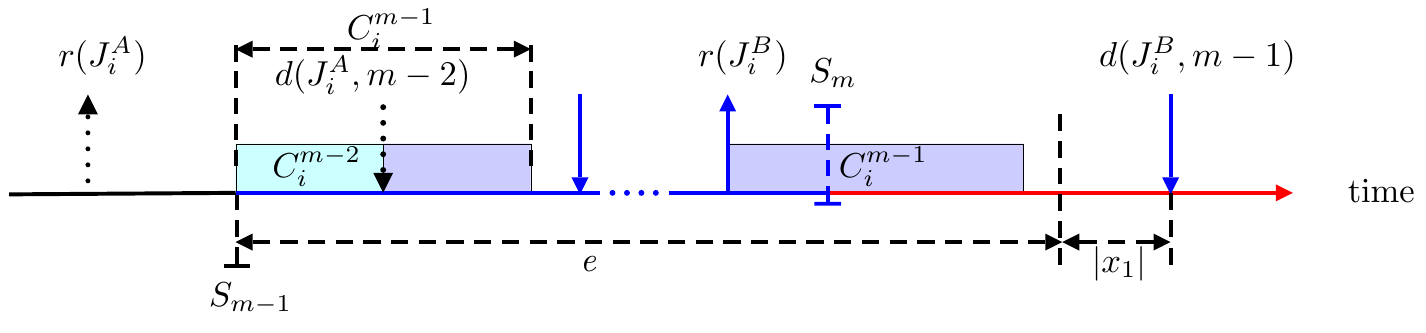} 
\caption{Case C4}
\label{fig:c4}
\end{figure}
As we shift the release pattern right, then  obviously the demand of $\tau_i$ would stay the same or decrease. On the other hand as we shift the release pattern left, i.e., $r(J_i^{A})= -D_i^{m-2}+C_i^{m-2}+x|x<0$,  demand of $J_i^B$ will first increase from $0$ to $y_1=dbf_{MM}(J_i^{B},S_m,\mathit e,m)|r(J_i^{B})+D_i^{m-1}=\mathit e$ (assuming at this time $x=x_1$), and then increases  linearly to $C_i^{m-1}$.
When $x> x_1$, the demand of $J_i^A$ will decrease but the demand of other jobs (including $J_i^B$) stay the same. When $x=x_1$, the change demand of $J_i^A$ is equal to $-C_i^{m-1}$ if $x_1<-C_i^{m-2}$, or $x_1$ if $x_1\geq -C_i^{m-2}$.  If  $x_1<-C_i^{m-2}$ and $y_1+x_1\geq 0$, then the demand of $\tau_i$ maximized at this time ($r(J_i^{A})=mod(\mathit e- D_i^{m-1},T_i)-T_i$) because as we further shift the pattern left, the the demand of $\tau_i$ would only decrease or stay the same.    Otherwise if $y_1+x_1<0$, then  total demand of $\tau_i$  maximizes demand when $r(J_i^{A})=-D_i^{m-2}+C_i^{m-2}$ because as we further shift the pattern left, the total demand would only stay the same or decrease.
  
In sum  the demand of $\tau_i|L_i= m-1$  during $[0,\mathit e)$ is maximized if $r(J_i^{A})=-D_i^{m-2}+C_i^{m-2}$ or $r(J_i^{A})=mod(\mathit e- D_i^{m-1},T_i)-T_i$.

\end{proof}

\begin{proof}[for Lemma~\ref{lem:pat3}]

\textbf{(C1):} If  $\mathit e-S_m>D_i^m$,  the demand of $\tau_i|L_i\geq m$  during $[0,\mathit e)$ is maximized when $r(J_i^{C})=mod(\mathit e- D_i^{m},T_i))-T_i$ (the last job released before $\mathit e$, $J_i^L$, has $d(J_i^L,m)=\mathit e$)  as shown in Figure~\ref{fig:c5}.
\begin{figure}[h!]
\centering
\includegraphics[scale=0.65]{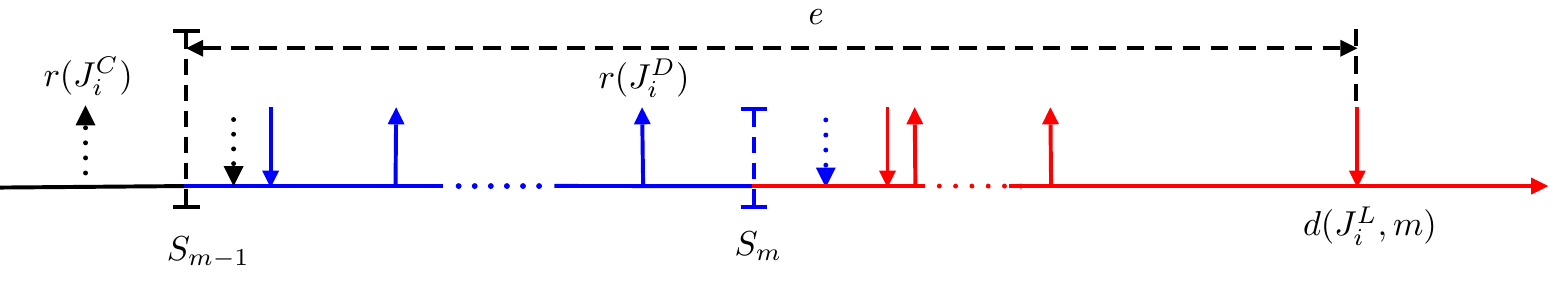} 
\caption{Case when $\mathit e-S_m>D_i^m$}
\label{fig:c5}
\end{figure}
This is because  as we shift the release pattern right, the demand of $J_i^L$ would decrease from $C_i^m$ to $0$ while the increase in demand of other jobs including $J_i^{C}$ and $J_i^{D}$  is bounded by $C_i^m$. On the other hand as we shift the release pattern left, the total demand of all jobs would decrease or stay the same.

\textbf{(C2):} If  $\mathit e-S_m<D_i^m-D_i^{m-1}$, no job of $\tau_i$ could execute more than $C_i^{m-1}$ (no job will execute after $S_m$). Thus $\tau_i$ behaves as a task with $L_i=m-1$, and  from Lemma~\ref{lem:pat2} we know the demand of $\tau_i$ is maximized  when $r(J_i^{C})\in\{-D_i^{m-2}+C_i^{m-2}, mod(\mathit e- D_i^{m-1},T_i)-T_i\}$.

\textbf{(C3):} If $D_i^m-D_i^{m-1} \leq \mathit e-S_m\leq D_i^m$,  at most one job (either $J_i^{C}$ or  $J_i^{D}$ ) could generate execution demand as much as $C_i^m$ as shown in Figure~\ref{fig:c6}.
\begin{figure}[h!]
\centering
\includegraphics[scale=0.7]{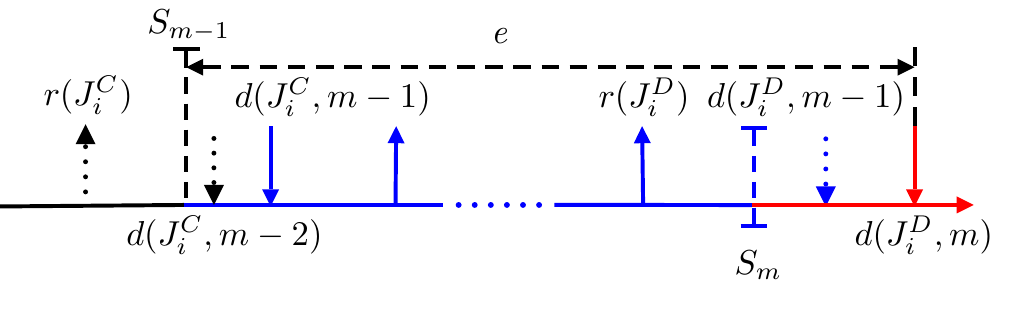} 
\caption{Case when $D_i^m-D_i^{m-1} \leq \mathit e-S_m\leq D_i^m$}
\label{fig:c6}
\end{figure} 
Thus suppose initially $r(J_i^C)=mod(\mathit e- D_i^{m},T_i))-T_i$, i.e., the scenario of Figure~\ref{fig:c6}. If we shift the release pattern left, the demand of $\tau_i$ would only decrease or stay the same. 
On the other hand as we shift the release pattern right, no job would execute after $S_m$, and hence $\tau_i$ behaves as a task with $L_i=m-1$. Therefore from Lemma~\ref{lem:pat2} we know the demand of $\tau_i$ is maximized  when $r(J_i^{C})\in\{-D_i^{m-2}+C_i^{m-2}, mod(\mathit e- D_i^{m-1},T_i)-T_i\}$.

In sum  the demand of $\tau_i|L_i\geq m$  during $[0,\mathit e)$ is maximized if $r(J_i^{C})=-D_i^{m-2}+C_i^{m-2}$ or $r(J_i^{C})=mod(\mathit e- D_i^{m-1},T_i)-T_i$ or $r(J_i^{C})=mod(\mathit e- D_i^{m},T_i)-T_i$.
\end{proof}

}

\subsection{Upper Bound of time interval length}
\label{sec:emax}
For task $\tau_i$ with $L_i=m-1$, we can observe that its demand during $[0,S_m)$ is upper bounded by  $(\frac{S_m-2\times C_i^{m-1}}{T_i}+2)\times C_i^{m-1}$. For task $\tau_i$ with $L_i\geq m$, its demand during $[0,\mathit e)$ is upper bounded by 
\begin{align*}
\small
\begin{split}
	 &\underbrace{C_i^{m-1}}_{\geq dbf_{MM}(J_i^{C},S_m,\mathit e,m)}+ \underbrace{\frac{S_m-C_i^{m-1}}{T_i}\times C_i^{m-1}}_{\geq n_{m-1}\times C_i^{m-1}}+\underbrace{C_i^m}_{\geq dbf_{MM}(J_i^{D},S_m,\mathit e,m)}+ \underbrace{\frac{\mathit e-S_m-D_i^m+T_i}{T_i}\times C_i^m}_{\geq n_m\times C_i^m}\\
	 &= C_i^{m-1} \times \frac{S_m-C_i^{m-1}+T_i}{T_i}+ \frac{\mathit e-S_m-D_i^m+2\times T_i}{T_i}\times C_i^m
\end{split}
\end{align*}
Suppose $dbf_{MM}(\tau,S_m,\mathit{e},m)>\mathit e$, then it must be that 
{\footnotesize\begin{align*}
&\sum\limits_{L_i\geq m} \left( C_i^{m-1}\! \times \!\frac{S_m\!-C_i^{m-1}\!+T_i}{T_i}
\!+\! \frac{\mathit e\!-\!S_m-\!D_i^m+\!2\times T_i}{T_i}\times C_i^m\right)
\!\!+\!\!\!\!\sum\limits_{L_i=m-1} \left (\frac{S_m+2\times(T_i- C_i^{m-1})}{T_i}\right)\!\times\! C_i^{m-1}\\
&=S_m\times \underbrace{\left(\sum\limits_{L_i=m-1}\frac{C_i^{m-1}}{T_i}+\sum\limits_{L_i\geq m}\frac{C_i^{m-1}-C_i^m}{T_i}    \right)}_{\mbox{Exp A} } + \underbrace{\sum\limits_{L_i=m-1}\frac{T_i-C_i^{m-1}}{T_i}\times 2\times C_i^{m-1}}_{\mbox{Exp B}}\\
&+\underbrace{\left(\sum\limits_{L_i\geq m}\frac{ 2\times T_i-D_i^m}{T_i} \times C_i^m+ \frac{T_i-C_i^{m-1}}{T_i}\times C_i^{m-1}\right)}_{\mbox{Exp C}}+ 
\sum\limits_{L_i\geq m}\frac{C_i^m}{T_i}\times  \mathit e>\mathit e\\
\end{align*}}
If $\mbox{Exp A}>0$, the total demand is maximized if $S_m=\mathit e$, and else if $\mbox{Exp A}\leq 0$, the total demand is maximized if $S_m=0$. Therefore the value of $\mathit e$ is bounded by 
$
	(\mbox{Exp B+ Exp C})/({1-\sum\limits_{L_i\geq m}\frac{C_i^m}{T_i}})$
or 
$
	(\mbox{Exp B+ Exp C})/({1-\sum\limits_{L_i\geq m}\frac{C_i^m}{T_i}-\mbox{Exp A}}) 	$


\bibliographystyle{ACM-Reference-Format-Journals}
\bibliography{all}








\end{document}